\documentclass[10pt,english,twocolumn,twoside]{IEEEtran}
\usepackage{mathalpha}
\usepackage[latin9]{inputenc}
\usepackage{float}
\usepackage{amsmath}
\usepackage{amssymb}
\usepackage{graphicx}
\usepackage{epstopdf}
\usepackage{color}
\usepackage{url}
\usepackage[nospace,compress]{cite}
\usepackage{amsbsy}
\usepackage{epsfig}
\usepackage{subcaption}

\usepackage[cal=BOONDOX]{mathalfa}

\makeatletter
\def\bx{{\mathbf x}}
\def\bc{{\mathbf c}}
\def\bv{{\mathbf v}}
\def\b0{{\mathbf 0}}

\newcommand{\beq}{\begin{equation}}
\newcommand{\eeq}{\end{equation}}

\def\bA{{\mathbf A}}

\def\bv{\mbox{\boldmath $v$}}

\def\bc{\mbox{\boldmath $c$}}
\def\bu{\mbox{\boldmath $u$}}

\def\by{\mbox{\boldmath $y$}}
\def\bv{\mbox{\boldmath $v$}}

\def\bx{\mbox{\boldmath $x$}}

\def\bs{\mbox{\boldmath $s$}}
\def\bt{\mbox{\boldmath $t$}}

\def\by{\mbox{\boldmath $y$}}

\def\blambda{\mbox{\boldmath $\lambda$}}

\def\bA{\mbox{\boldmath $A$}}

\def\mB{\mbox{$\mathbf{B}$}}

\def\mD{\mbox{$\mathbf{D}$}}

\def\mX{\mbox{$\mathbf{X}$}}
\def\mD{\mbox{$\mathbf{D}$}}

\def\mG{\mbox{$\mathbf{G}$}}
\def\mF{\mbox{$\mathbf{F}$}}

\def\mI{\mbox{$\mathbf{I}$}}
\def\mL{\mbox{$\mathbf{L}$}}

\def\mP{\mbox{$\mathbf{P}$}}
\def\mQ{\mbox{$\mathbf{Q}$}}

\def\mR{\mbox{$\mathbf{R}$}}
\def\mS{\mbox{$\mathbf{S}$}}

\def\mU{\mbox{$\mathbf{U}$}}
\def\mV{\mbox{$\mathbf{V}$}}

\def\mM{\mbox{$\mathbf{M}$}}

\newcommand{\ds}{\displaystyle}

\def\V{{\cal V}}
\def\E{{\cal E}}
\def\T{{\cal T}}

\def\x{{\mathbf x}}

\def\A{{\cal A}}
\newtheorem{theorem}{\textbf{Theorem}}
\newtheorem{proposition}{Proposition}
\newtheorem{definition}{Definition}

\newenvironment{proof}[1][Proof]{\noindent \textbf{#1.} }{\qedsymbol}
\newcommand{\qedsymbol}{\hspace{\fill}\rule{1.5ex}{1.5ex}}
\floatstyle{ruled}
\newfloat{algorithm}{tbp}{loa}
\providecommand{\algorithmname}{Algorithm}
\floatname{algorithm}{\protect\algorithmname}

\usepackage{amssymb}
\begin{document}

\title{Topological Signal Processing over\\ Simplicial Complexes}
\author{Sergio Barbarossa,~\IEEEmembership{Fellow,~IEEE}, Stefania Sardellitti,~\IEEEmembership{Member,~IEEE} \\
\thanks{The authors are with the Department of Information
Engineering, Electronics, and Telecommunications,
Sapienza University of Rome, Via Eudossiana 18, 00184,
Rome, Italy. E-mails: \{sergio.barbarossa; stefania.sardellitti\}@uniroma1.it.
This work has been supported by
H2020 EU/Taiwan Project 5G CONNI, Nr. AMD-861459-3 and  by MIUR, under the PRIN Liquid-Edge contract.
Some preliminary results of this work were presented at the 2018 IEEE Workshop in Data Science \cite{barbarossa2018learning}.
}}

\maketitle

\begin{abstract}
The goal of this paper is to establish the fundamental tools to analyze signals defined over a topological space, i.e. a set of points along with a set of neighborhood relations. This setup does not require the definition of a metric and then it is especially useful to deal with signals defined over non-metric spaces. We focus on signals defined over simplicial complexes. Graph Signal Processing (GSP) represents a special case of Topological Signal Processing (TSP), referring to the situation where the signals are associated only with the vertices of a graph. Even though the theory can be applied to signals of any order, we focus on signals defined over the edges of a graph and show how building a simplicial complex of order two, i.e. including triangles, yields benefits in the analysis of edge signals. After reviewing the basic principles of algebraic topology, we derive a sampling theory for signals of any order and emphasize the interplay between signals of different order. Then we propose a method to infer the topology of a simplicial complex from  data. We conclude with applications to real edge signals and to the analysis of discrete vector fields to illustrate the benefits of the proposed methodologies.
\end{abstract}
\section{Introduction}
\label{sec:intro}
Historically, signal processing has been developed for signals defined over a metric space, typically time or space.  More recently, there has been a surge of interest to deal with signals that are not necessarily defined over a metric space. Examples of particular interest are biological networks, social networks, etc.
The field of graph signal processing (GSP) has recently emerged as a framework to analyze signals defined over the vertices of a graph \cite{shuman2013}, \cite{ortega2018graph}. A graph ${\cal G}({\cal V}, {\cal E})$ is a simple example of topological space, composed of a set of elements (vertices)  ${\cal V}$ and a set of edges ${\cal E}$ representing {\it pairwise} relations. However, notwithstanding their enormous success, graph-based representations are not always able to capture all the information present in interconnected systems in which the complex interactions among the system constitutive elements cannot be reduced to pairwise interactions, but require a {\it multiway} description, as suggested in \cite{klamt2009hypergraphs, courtney2016generalized, giusti2016two, shen2018genome, benson2018simplicial, agarwal2006higher}.

To establish a general framework to deal with complex interacting systems, it is useful to start from
a topological space, i.e.  a set of elements  ${\cal V}$, along with an ensemble of {\it multiway} relations, represented by a set ${\cal S}$ containing subsets of various cardinality, of the elements of ${\cal V}$. The structure ${\cal H}({\cal V}, {\cal S})$ is known as a {\it hypergraph}. In particular, a class of hypergraphs that is particularly appealing for its rich algebraic structure is given by {\it simplicial complexes}, whose defining feature is the inclusion property stating that if a set ${\cal A}$ belongs to ${\cal S}$, then all subsets of  ${\cal A}$ also belong to  ${\cal S}$ \cite{munkres2000topology}. Restricting the attention to simplicial complexes is a limitation with respect to a hypergraph model. Nevertheless, simplicial complex models include many cases of interest and, most important, their rich algebraic structure makes possible to derive tools that are very useful for analyzing signals over the complex.
Learning simplicial complexes representing the complex interactions among sets of elements has been already proposed in brain network analysis \cite{giusti2016two}, neuronal morphologies \cite{kanari2018topological}, co-authorship networks \cite{patania2017shape}, collaboration networks \cite{ramanathan2011beyond}, \cite{moore2012analyzing}, tumor progression analysis \cite{roman2015simplicial}. More generally, the use of algebraic topology tools for the extraction of information from data is not new: The framework known as \textit{Topological Data Analysis} (TDA), see e.g. \cite{carlsson2009topology}, has exactly this goal. Interesting applications of algebraic topology tools have been proposed to control systems \cite{muhammad2006control}, statistical ranking from incomplete data \cite{jiang2011}, \cite{Xu}, distributed coverage control of sensor networks \cite{tahbaz2010, chintakunta2014distributed, desilva2007coverage}, wheeze detection \cite{emrani2014persistent}.  One of the fundamental tools of TDA is the analysis of persistent homologies extracted from data \cite{edelsbrunner2008persistent}, \cite{horak2009persistent}. Topological methods to analyze signals and images are also the subjects of the two books \cite{krim2015geometric} and \cite{robinson2016topological}.

The goal of our paper is to establish a fundamental framework to analyze signals defined over a simplicial complex. Our approach is complementary to TDA: Rather than focusing, like TDA, on the properties of the simplicial complex extracted from data, {\it we focus on the  properties of signals defined over a simplicial complex}. Our approach includes GSP as a particular case: While GSP focuses on the analysis of signals defined over the vertices of a graph, topological signal processing (TSP) considers signals defined over simplices of various order, i.e. signals defined over nodes, edges, triangles, etc.  Relevant examples of edge signals are flow signals, like blood flow between different areas of the brain \cite{huang2018graph}, data traffic over communication links \cite{leung1994traffic}, regulatory signals in gene regulatory networks \cite{sever2015signal}, where it was shown that a dysregulation of these regulatory signals is one of the causes of cancer \cite{sever2015signal}. Examples of signals defined over triplets are co-authorship networks, where a (filled) triangle indicates the presence of papers co-authored by the authors associated to its three
vertices \cite{patania2017shape} and the associated signal value counts the number of such publications. Further examples of even higher order structures are analyzed in \cite{benson2018simplicial}, with the goal of predicting higher-order links. There are previous works dealing with the analysis of edge signals, like \cite{Segarra_018, Segarra_019, Evans, Ahn}. More specifically, in \cite{Segarra_018} the authors introduced a class of filters  to analyze  edge signals based on the edge-Laplacian matrix \cite{mesbahi2010graph}. A semi-supervised learning method for learning edge flows was also suggested in \cite{Segarra_019}, where the authors proposed filters highlighting both divergence-free and curl-free behaviors.
Other works analyzed edge signals using a line-graph transformation  \cite{Evans}, \cite{Ahn}.
Random walks evolving over simplicial complexes have been  analyzed in  \cite{mukherjee2016random},\cite{parzanchevski2017simplicial}, \cite{schaub2018random}.
In \cite{mukherjee2016random},  random walks and diffusion process over simplicial complexes are introduced, while
in \cite{schaub2018random}  the authors focused on the study of diffusion processes on the edge-space by generalizing the well-known relationship between the normalized graph Laplacian operator and random walks on graphs.

Building a representation based on a simplicial complex is a straightforward generalization of a graph-based representation. Given a set of time series, it is well known that graph-based representations are very useful to capture relevant correlations or causality relations present between different signals \cite{mateos2019connecting},\cite{dong2019learning}. In such a case, each time series is associated to a node of a graph, and the presence of an edge is an indicator of the relations between the signals associated to its endpoints. As a direct generalization, if we have signals defined over the edges of a graph, capturing their relations requires inferring the presence of triangles associating triplets of edges.

\noindent In summary, our main contributions are listed below:
\begin{enumerate}
\item \it we show how to derive a real-valued function capturing triple-wise relations among data, to be used as a regularization function in the analysis of edge signals;
\item we derive a {\it sampling theory} defining the conditions for the recovery of high order signals from a subset of observations, highlighting the {\it interplay between signals of different order};
\item we propose {\it inference algorithms to extract the structure of the simplicial complex from high order signals};
\item  \it we apply our edge signal processing tools to the analysis of vector fields, with a specific attention to the recovery of the  RiboNucleic Acid (RNA) velocity vector field, useful to predict the evolution of a cell behavior \cite{la2018rna}.
\end{enumerate}
We presented some preliminary results of our work in \cite{barbarossa2018learning}. Here, we extend the work of \cite{barbarossa2018learning}, deriving a sampling theory for signals defined over complexes of various order,  proposing new inference methods, more robust against noise, and showing applications to the analysis of wireless data traffic and to the estimation of the RNA velocity vector field.

The paper is organized as follows.  Section \ref{sec:discrete calculus} recalls the main algebraic principles that will form the basis  for the derivation of the signal processing tools carried out in the ensuing sections. In Section \ref{sec:spectr_theory}, we will first recall the eigenvectors properties of higher-order Laplacian and the Hodge decomposition. Then, we derive the real-valued extension of an edge set function, capturing triple-wise relations among the elements of a discrete set, which will be used to design unitary bases to represent edge signals. In Section \ref{sec:edge_flows_est},
we  illustrate some methods to recover the edge signal components from noisy observations. Section \ref{sec:sampling} provides  theoretical  conditions to recover an edge signal
from a subset of samples, highlighting the interplay between signals defined over structures of different order. In Section \ref{Estimation of discrete vector fields}, we illustrate how to use edge signal processing to filter discrete vector fields. Then, in Section \ref{sec:L1_inference}, we propose some methods to infer the simplicial complex structure from  noisy observations, illustrating their performance over both synthetic and real data.
Finally, in Section \ref{sec:conclusions} we draw some conclusions.

\section{Review of algebraic topology tools}
\label{sec:discrete calculus}
In this section we recall the basic principles of algebraic topology \cite{munkres2000topology} and discrete calculus \cite{grady2010}, as they will form the background required for deriving the basic signal processing tools to be used in later sections. We follow an algebraic approach that is accessible also to readers  with no specific background on algebraic topology.

\subsection{Discrete domains: Simplicial complexes}
Given a finite set $\mathcal{V}\triangleq \{ v_0,\ldots, v_{N-1}\}$ of $N$ points (vertices), a {\it $k$-simplex} $\sigma^k_i$ is an unordered set $\{v_{i_0}, \ldots, v_{i_k}\}$ of $k+1$ points with $0 \le i_j \le N-1$, for $j=0,\ldots, k$, and $v_{i_j}\neq v_{i_n}$ for all $i_j \neq i_n$. A \textit{face} of the $k$-simplex $\{v_{i_0}, \ldots, v_{i_k}\}$
is a $(k-1)$-simplex of the form $\{v_{i_0}, \ldots, v_{i_{j-1}}, v_{i_{j+1}},\ldots, v_{i_k}\}$,
for some $0\le j \le k$. Every $k$-simplex has exactly $ k+1$ faces. An \textit{abstract simplicial complex} ${\cal X}$ is a finite collection of simplices that is closed under inclusion of faces, i.e., if ${\cal \sigma}_i \in {\cal X}$, then all faces of $\sigma_i$ also belong to ${\cal X}$.
The order (or dimension) of a simplex is one less than its cardinality. Then, a vertex is a 0-dimensional simplex, an edge  has dimension $1$, and so on. The dimension of a simplicial complex is the largest dimension of any of its simplices.
A graph is a particular case of an abstract simplicial complex of order $1$, containing only simplices of order $0$ (vertices) and $1$ (edges).

If the set of points is embedded in a real space $\mathbb{R}^D$  of dimension $D$, we can associate a {\it geometric simplicial complex} with  the abstract complex. A set of points in a real space $\mathbb{R}^D$  is \textit{affinely independent} if it is not contained in a hyperplane; an affinely independent set in $\mathbb{R}^D$ contains at most $D + 1$ points. A geometric  $k$-simplex is the \textit{convex hull} of a set of $k + 1$ affinely independent points, called its vertices. Hence, a point is a $0$-simplex, a line segment is a $1$-simplex, a triangle is a $2$-simplex,  a tetrahedron is a $3$-simplex, and so on. A geometric simplicial complex shares the fundamental property of an abstract simplicial complex: It is a collection of simplices that is closed under inclusion and with the further property that  the intersection of any two simplices in ${\cal X}$ is also a simplex  in ${\cal X}$, assuming that the empty set  is an element of every simplicial complex. Although geometric simplicial complexes are easier to visualize and, for this reason, we will often use geometric terms like edges, triangles, and so on, as synonyms of pairs, triplets, we do not require the simplicial complex to be embedded in any real space, so as to leave the treatment as general as possible.

The structure of a simplicial complex is captured by the neighborhood relations of its subsets.
As with graphs, it is useful to introduce first the orientation of the simplices. Every simplex, of any order, can have only two orientations,  depending on the permutations of its elements. Two orientations are equivalent, or coherent, if each of them can be recovered from the other by an even number of transpositions, where each transposition is a permutation of two elements \cite{munkres2000topology}. A $k$-simplex $\sigma_i^k\equiv\{v_{i_0}, v_{i_1}, \ldots, v_{i_k}\}$ of order $k$, together with an orientation is an {\it oriented} $k$-simplex and is denoted by $\left[v_{i_0}, v_{i_1}, \ldots, v_{i_k}\right]$. Two simplices of order $k$, $\sigma_i^k, \sigma_j^k \in {\cal X}$, are {\it upper adjacent} in ${\cal X}$, if both are faces of a simplex of order $k+1$. Two simplices of order $k$, $\sigma_i^k, \sigma_j^k \in {\cal X}$, are {\it lower adjacent} in ${\cal X}$, if both have a common face of order $k-1$ in ${\cal X}$.
A  $(k-1)$-face $\sigma_j^{k-1}$ of a $k$-simplex $\sigma_i^{k}$ is called a boundary element of $\sigma_i^{k}$. We use the notation $\sigma_j^{k-1}\subset \sigma_i^{k}$ to indicate that $\sigma_j^{k-1}$ is a boundary element of $\sigma_i^{k}$. Given a simplex $\sigma_j^{k-1} \subset \sigma_i^{k}$, we use the notation $ \sigma_j^{k-1} \sim \sigma_i^{k}$ to indicate that the orientation of $\sigma_j^{k-1}$ is coherent with that of $\sigma_i^{k}$, whereas we write $ \sigma_j^{k-1} \nsim \sigma_i^{k}$ to indicate that the two orientations are opposite.

For each $k$,  $C_k({\cal X}, \mathbb{R})$ denotes the vector space obtained by the linear combination, using real coefficients, of the set of oriented $k$-simplices of ${\cal X}$. In algebraic topology, the elements of $C_k({\cal X}, \mathbb{R})$ are called $k$-\textit{chains}.  If $\{\sigma_1^k, \ldots, \sigma_{n_k}^{k}\}$ is the set of $k$-simplices in ${\cal X}$, a $k$-chain $\tau_k$ can be written as $\tau_k=\sum_{i=1}^{n_k} \alpha_i \sigma_i^k$. Then, given the basis $\{\sigma_1^k, \ldots, \sigma_{n_k}^{k}\}$, a chain $\tau_k$ can be represented by the vector of its expansion coefficients $(\alpha_1, \ldots, \alpha_{n_k})$.
An important operator acting on ordered chains is the \textit{boundary operator}. The boundary of the ordered $k$-chain  $[v_{i_0}, \ldots, v_{i_k}]$ is a linear mapping $\partial_k : C_k({\cal X},\mathbb{R})\rightarrow  C_{k-1}({\cal X},\mathbb{R})$ defined as
\begin{equation}
\label{boundary}
\partial_k  [v_{i_0}, \ldots, v_{i_k}] \triangleq \sum_{j=0}^{k}\, (-1)^j [v_{i_0}, \ldots, v_{i_{j-1}}, v_{i_{j+1}}, \ldots, v_{i_k}].
\end{equation}
So, for example, given an oriented triangle $\sigma^2_i\triangleq [v_{i_0}, v_{i_1}, v_{i_2}]$, its boundary is
\begin{equation}
\label{boundary_triangle}
\partial_2 \sigma^2_i=[v_{i_1}, v_{i_2}]-[v_{i_0}, v_{i_2}]+[v_{i_0}, v_{i_1}],
\end{equation}
i.e., a suitable linear combination of its edges.
It is straightforward to verify, by simple substitution, that {\it the boundary of a boundary is zero}, i.e., $\partial_k \partial_{k+1} = 0$.

It is important to remark that an oriented simplex is different from a directed one. As with graphs, an oriented edge establishes which direction of the flow is considered positive or negative, whereas a directed edge only permits flow in one direction \cite{grady2010}. In this work we will considered oriented, undirected simplices.

\subsection{Algebraic representation}
The structure of a simplicial complex ${\cal X}$ of dimension $K$,  shortly named  $K$-simplicial complex, is fully described by the set of its incidence matrices $\mathbf{B}_k$, $k=1, \ldots, K$. Given an orientation of the simplicial complex ${\cal X}$, the entries of the incidence matrix $\mathbf{B}_k$ establish which $k$-simplices are incident to which $(k-1)$-simplices. Then $\mathbf{B}_k$
is the matrix representation of the boundary operator. Formally speaking, its entries are defined as follows:
  \beq \label{inc_coeff}
  B_k(i,j)=\left\{\begin{array}{rll}
  0, & \text{if} \; \sigma^{k-1}_i \not\subset \sigma^{k}_j \\
  1,& \text{if} \; \sigma^{k-1}_i \subset \sigma^{k}_j \;  \text{and} \; \sigma^{k-1}_i \sim \sigma^{k}_j\\
  -1,& \text{if} \; \sigma^{k-1}_i \subset \sigma^{k}_j \;  \text{and} \; \sigma^{k-1}_i \nsim \sigma^{k}_j\\
  \end{array}\right. .
  \eeq
If we consider, for simplicity, a simplicial complex of order two, composed of a set $\V$ of vertices, a set $\E$ of edges, and a set $\T$ of triangles,  having cardinalities $V=|\V|$,
$E=|\E|$, and $T=|\T|$, respectively, we need to build two incidence matrices $\mB_1 \in \mathbb{R}^{V \times E}$ and  $\mB_2 \in \mathbb{R}^{E \times T}$.

From (\ref{boundary}), the property that the boundary of a boundary is zero maps into the following matrix form
\begin{equation}
\label{boundary of boundary}
\mathbf{B}_k \mathbf{B}_{k+1}=\mathbf{0}.
\end{equation}
The structure of a $K$-simplicial complex is fully described by its {\it high order combinatorial Laplacian} matrices \cite{goldberg2002combinatorial}, of order $k=0, \ldots, K$, defined as
\begin{align}
\label{Laplacians}
&\mathbf{L}_0=\mathbf{B}_{1}\mathbf{B}_{1}^T,\\
&\mathbf{L}_k=\mathbf{B}^T_k\mathbf{B}_k+\mathbf{B}_{k+1}\mathbf{B}_{k+1}^T, \, k=1, \ldots, K-1, \label{eq:L_1}\\
&\mathbf{L}_K=\mathbf{B}_{K}^T\mathbf{B}_{K}. \label{eq:L_K}
\end{align}
It is worth emphasizing that all Laplacian matrices of intermediate order, i.e. $k=1, \ldots, K-1$, contain two terms: The first term, also known as {\it lower Laplacian}, expresses the lower adjacency of $k$-order simplices; the second terms, also known as {\it upper Laplacian}, expresses the upper adjacency
of $k$-order simplices. So, for example, two edges are lower adjacent if they share a common vertex, whereas they are upper adjacent if they are faces of a common triangle. Note that the vertices of a graph can only be upper adjacent, if they are incident to the same edge. This is why the graph Laplacian $\mathbf{L}_0$ contains only one term.

\section{Spectral simplicial theory}
\label{sec:spectr_theory}
In this paper, we are interested in analyzing signals defined over a simplicial complex. Given a set  ${\cal S}_k$, with elements of cardinality $k+1$, a signal is defined as a real-valued map  on the elements of ${\cal S}_k$, of the form
\begin{equation}
f_{k}: {\cal S}_k \rightarrow \mathbb{R}, \,\, k=0, 1, \ldots
\end{equation}
The order of the signal is one less the cardinality of the elements of ${\cal S}_k$.
Even though our framework is general, in many cases we focus on simplices of order up to two. In that case, we consider a set of vertices $\V$, a set of edges $\E$ and a set of triangles $\T$, of dimension $V$, $E$, and $T$, respectively. We denote with $\cal X (\V, \E, T)$ the associated simplicial complex. The signals over each complex of order $k$, with $k=0, 1$ and $2$, are defined as the following maps: $\bs^0: {\cal V} \rightarrow \mathbb{R}^V$, $\bs^1: {\cal E} \rightarrow \mathbb{R}^E$, and $\bs^2: {\cal T} \rightarrow \mathbb{R}^T$.

Spectral graph theory represents a solid framework to extract features of a graph looking at the eigenvectors of the combinatorial Laplacian  $\mathbf{L}_0$ of order $0$.
The eigenvectors associated with the smallest eigenvalues of $\mathbf{L}_0$ are very useful, for example, to identify clusters \cite{von2007tutorial}. Furthermore, in GSP it is well known that a suitable basis to represent signals defined over the vertices of a graph, i.e. signals of order $0$, is given by the eigenvectors of $\mathbf{L}_0$. In particular, given the eigendecomposition of $\mathbf{L}_0$:
\begin{equation}
\mathbf{L}_0=\mathbf{U}_0 \mathbf{\Lambda}_0\mathbf{U}_0^T,
\end{equation}
the Graph Fourier Transform (GFT) of a signal $\bs^0$ over an undirected graph has been defined as the projection of the signal onto the space spanned by the eigenvectors of $\mathbf{L}_0$, i.e. (see  \cite{ortega2018graph}
and the references therein)
\begin{equation}
\label{s=U s^tilde_1}
\widehat{\bs}^0 \triangleq \mathbf{U}_0^T\, \bs^0.
\end{equation}
Equivalently, a signal defined over the vertices of a graph can be represented as
\begin{equation}
\label{s=U s^tilde}
\bs^0=\mathbf{U}_0\, \widehat{\bs}^0.
\end{equation}
From graph spectral theory, it is well known that the eigenvectors associated with the smallest eigenvalues of $\mathbf{L}_0$ encode information about the clusters of the graph. Hence, the representation given by (\ref{s=U s^tilde}) is particularly suitable for signals that are smooth within each cluster, whereas they can vary arbitrarily across different clusters. For such signals, in fact, the representation in  (\ref{s=U s^tilde}) is {\it sparse} or approximately sparse.

As a generalization of the above approach, we may  represent signals of various order over bases built with the eigenvectors of the corresponding high order Laplacian matrices, given in (\ref{eq:L_1})-(\ref{eq:L_K}). Hence, using the eigendecomposition
\begin{equation}
\mathbf{L}_k=\mathbf{U}_k \mathbf{\Lambda}_k\mathbf{U}_k^T,
\end{equation}
we may define the GFT of order $k$ as the projection of a $k$-order signal onto the eigenvectors of $\mathbf{L}_k$, i.e.
\begin{equation}
\label{k-GFT}
\widehat{\bs}^k \triangleq \mathbf{U}_k^T\, \bs^k,
\end{equation}
so that a signal $\bs^k$ can be represented in terms of its GFT coefficients as
\begin{equation}
\label{s^k=U_k s^tilde^k}
\bs^k=\mathbf{U}_k\, \widehat{\bs}^k.
\end{equation}
Now we want to show  under what conditions (\ref{s^k=U_k s^tilde^k}) is a meaningful representation of a $k$-order signal and what is the meaning of such a representation. More specifically, the goal of this section is threefold: i) we recall first the relations between eigenvectors of various order of $\mathbf{L}_k$; ii) we recall the Hodge decomposition, which is a basic theory showing that the eigenvectors of any order can be split into three different classes, each representing a specific behavior of the signal; iii) we provide a theory showing how to build a {\it topology-aware} unitary basis to represent signals of various order starting only from topological properties.
\subsection{Relations between eigenvectors of different order}
There are interesting relations between the eigenvectors of Laplacian matrices of different order \cite{Horak13}, \cite{steenbergen2013towards} which is useful to recall as they play a key role in spectral analysis. The following properties hold true for the eigendecomposition of $k$-order Laplacian matrices, with $k=1, \ldots, K-1$.\\

\begin{proposition} Given the  Laplacian matrices $\mathbf{L}_k$ of any order $k$, with $k=1, \ldots, K-1$, it holds:
\begin{enumerate}
\item the eigenvectors associated with the nonzero eigenvalues of $\mathbf{B}^T_k\mathbf{B}_k$ are orthogonal to the eigenvectors associated with the nonzero eigenvalues of  $\mathbf{B}_{k+1}\mathbf{B}_{k+1}^T$ and viceversa;
\item if $\bv$ is an eigenvector of $\mathbf{B}_{k}\mathbf{B}_{k}^T$ associated with the eigenvalue $\lambda$, then $\mathbf{B}_{k}^T \bv$ is an eigenvector of $\mathbf{B}^T_k\mathbf{B}_k$, associated with the same eigenvalue;
\item the eigenvectors  associated with the nonzero eigenvalues $\lambda$ of  $\mathbf{L}_k$ are either the eigenvectors of $\mathbf{B}^T_k\mathbf{B}_k$ or those of $\mathbf{B}_{k+1}\mathbf{B}_{k+1}^T$;
\item the  nonzero eigenvalues of $\mathbf{L}_k$ are either the eigenvalues of $\mathbf{B}^T_k\mathbf{B}_k$ or those of $\mathbf{B}_{k+1}\mathbf{B}_{k+1}^T$.
\end{enumerate}
\end{proposition}
\begin{proof}
All above properties are easy to prove. Property 1) is straightforward: If $\mathbf{B}_k^T \mathbf{B}_k \bv=\lambda \bv$, then
\begin{equation}
\mathbf{B}_{k+1}\mathbf{B}_{k+1}^T \lambda \bv= \mathbf{B}_{k+1}\mathbf{B}_{k+1}^T \mB_k^T \mB_k \bv  = \b0
\end{equation}
because of (\ref{boundary of boundary}). Similarly, for the converse.
Property 2) is also straightforward: If $\bv$ is an eigenvector of $\mathbf{B}_k \mathbf{B}_k^T$ associated with a nonvanishing eigenvalue $\lambda$, then
\begin{equation} \label{eq:pro_two}
(\mathbf{B}_k^T \mathbf{B}_k) \mathbf{B}_k^T \bv=\mathbf{B}_k^T (\mathbf{B}_k \mathbf{B}_k^T) \bv =\lambda \mathbf{B}_k^T \bv.
\end{equation}
Finally, properties 3) and 4) follow from the definition of $k$-order Laplacian, i.e.   $\mathbf{L}_k=\mathbf{B}^T_k\mathbf{B}_k+\mathbf{B}_{k+1}\mathbf{B}_{k+1}^T$ and from property 1).
\end{proof}

\noindent{\bf Remark}: Recalling that the eigenvectors associated with the smallest nonzero eigenvalues of $\mathbf{L}_0$ are smooth within each cluster, applying property 2) to the case $k=1$, it turns out that the eigenvectors of $\mathbf{L}_1$ associated with the smallest eigenvalues of $\mathbf{B}_{1}^T\mathbf{B}_{1}$ are approximately null over the links within each cluster, whereas they assume the largest (in modulus) values on the edges across clusters. These eigenvectors are  then useful to highlight {\it inter-cluster edges}.

\subsection{Hodge decomposition}
Let us  consider the eigendecomposition of the $k$-th order Laplacian, for $k=1, \ldots, K-1$,
\begin{equation}
\mathbf{L}_k=\mathbf{B}^T_k\mathbf{B}_k+\mathbf{B}_{k+1}\mathbf{B}_{k+1}^T=\mathbf{U}_k \mathbf{\Lambda}_k \mathbf{U}_k^T.
\end{equation}
The structure of $\mathbf{L}_k$, together with the property $\mathbf{B}_k\mathbf{B}_{k+1}=\mathbf{0}$, induces an interesting decomposition of the space $\mathbb{R}^{D_k}$ of signals of order $k$ of dimension ${D_k}$. First of all, the property $\mathbf{B}_k\mathbf{B}_{k+1}=\mathbf{0}$ implies
${\rm img}(\mathbf{B}_{k+1})\subseteq{\rm ker}(\mathbf{B}_{k})$. Hence, each vector $\bx \in  {\rm ker}(\mathbf{B}_k)$ can be decomposed into two parts: one belonging to ${\rm img}(\mathbf{B}_{k+1})$ and one orthogonal to it. Furthermore,  since the whole space $\mathbb{R}^{D_k}$ can always be    written as $\mathbb{R}^{D_k}\equiv {\rm ker}(\mathbf{B}_k)  \oplus {\rm img}(\mathbf{B}_k^T)$,  it is possible to decompose  $\mathbb{R}^{D_k}$ into the direct sum \cite{Lim}
\beq \label{eq:Hodge}
\begin{split}
\mathbb{R}^{D_k} &\equiv \text{img}(\mathbf{B}_{k}^{T}) \oplus \text{ker}(\mathbf{L}_k) \oplus \text{img}(\mathbf{B}_{k+1}),
\end{split}
\eeq
where the vectors in $\text{ker}(\mathbf{L}_k)$ are also in $\text{ker}(\mathbf{B}_k)$ and $\text{ker}(\mathbf{B}_{k+1}^T)$.
This implies that, given any signal $\bs^k$ of order $k$, there always exist three signals $\bs^{k-1}$, $\bs_{H}^k$, and $\bs^{k+1}$, of order $k-1$, $k$, and $k+1$, respectively, such that $\bs^k$ can always be expressed as the sum of three {\it orthogonal} components:
\begin{equation}
\label{s_decomp}
    \bs^k=\mathbf{B}_{k}^T\, \bs^{k-1}+\bs_{H}^k +\mathbf{B}_{k+1}\, \bs^{k+1}.
\end{equation}
This decomposition is known as  \textit{Hodge decomposition} \cite{eckmann1944harmonische} and it is the extension of the Hodge theory for differential forms on Riemannian manifold to  simplicial complexes. The subspace $\text{ker}(\mathbf{L}_k)$ is  called harmonic subspace since each
$\bs_{H}^k \in \text{ker}(\mathbf{L}_k)$ is a solution of the {\it discrete} Laplace equation $$\mathbf{L}_k\,\bs_{H}^k= (\mathbf{B}_k^T \mathbf{B}_k + \mathbf{B}_{k+1} \mathbf{B}_{k+1}^T) \, \bs_{H}^k=\mathbf{0}.$$
When embedded in a real space, a fundamental property of geometric simplicial complexes of order $k$ is that the dimensions of $\text{ker}(\mathbf{L}_k)$, for $k=0, \ldots, K$ are \textit{topological invariants} of the $K$-simplicial complex, i.e. topological features that are preserved under homeomorphic transformations of the space. The dimensions of $\text{ker}(\mathbf{L}_k)$ are also known as {\it Betti numbers} $\beta_k$ of order $k$:  $\beta_0$ is the number of connected components of the graph, $\beta_1$ is the number of holes, $\beta_2$ is the number of cavities, and so on \cite{eckmann1944harmonische}.

The decomposition in (\ref{s_decomp}) shows an interesting interplay between signals of different order, which we will exploit in the ensuing sections. Before proceeding, it is useful to clarify the meaning of the terms appearing in (\ref{s_decomp}).  Let us consider, for simplicity, the analysis of flow signals, i.e. the case $k=1$.
To this end, it is useful to introduce the curl and divergence operators, in analogy with their continuous time counterpart operators applied to vector fields. More specifically, given an edge signal $\bs^1$, the (discrete) curl operator is defined as
\begin{equation}
    \mbox{curl}(\bs^1)=\mB_2^T \bs^1.
\end{equation}
This operator maps the edge signal $\bs^1$ onto a signal defined over the triangle sets, i.e. in $\mathbb{R}^T$, and it is straightforward to verify that the generic $i$-th entry of the $\mbox{curl}(\bs^1)$ is a measure  of the {\it flow circulating along the edges of the $i$-th triangle}.
As an example for the simplicial complex in Fig. \ref{cut}, we  have
\beq \label{B2_ex}
\mB_2^T=\left[\begin{array}{cccccccccccc}
0 & 0 &1 & \!\!\!\!-1 & 1 & 0 & 0& 0& 0& 0& 0\\
0 & 1 & \!\!\!\!-1 & 0 & 0 & 0 & 1 & 0 & 0 & 0 & 0\\
0 & 0 & 0 & 0 & 0 & 0 &\!\!\!\!-1 & 1 & 1 & 0 & 0
\end{array}\right]
\eeq
and, defining  $\bs^1=[e_1,e_2,\ldots,e_{10}, e_{11}]^T$, we get
 $\mbox{curl}(\bs^1)=[e_3-e_4+e_5, e_2-e_3+e_7,-e_7+e_8+e_9]^T$. Each entry of $\bs^1$ is then the circulation over the corresponding triangle.

Similarly, the (discrete) divergence operator maps the edge signal $\bs^1$ onto a signal defined over the vertex space, i.e.  ${\mathbb R}^V$, and it is defined as
\begin{equation}
    \mbox{div}(\bs^1)=\mB_1 \bs^1.
\end{equation}
Again, by direct substitution, it turns out that the $i$-th entry of $\mbox{div}(\bs^1)$ represents the {\it net-flow passing through the $i$-th vertex}, i.e. the difference between the inflow and outflow at node $i$. Thus a non-zero divergence reveals the presence of a source or sink node.
For the example in Fig. \ref{cut}, we  get
\beq
\mB_1=\left[\begin{array}{cccccccccccccccc}
\!\!\!\!-1 & 0 & 0 & 0 & 0 & \!\!\!\!-1 & 0& 0& 0& 0& 0\\
1 & \!\!\!\!-1 & \!\!\!\!-1 & \!\!\!\!-1 & 0 & 0 & 0 & 0 & 0 & 0 & 0\\
0 & 1 & 0 & 0 & 0 & 0 &\!\!\!\!-1 & 0 & \!\!\!\!-1 &\!\!\!\! -1 & 0 \\
0 & 0 & 0 & 0 & 0 & 0 & 0 & 0 & 0 & 1 & \!\!\!\!-1 \\
0 & 0 & 0 & 0 & 0 & 0 & 0 &\!\!\!\! -1 & 1 & 0 & 1\\
0 & 0 & 1 & 0 & \!\!\!\!-1 & 0 & 1 & 1 & 0 & 0 & 0\\
0 & 0 & 0 & 1 & 1 & 1 & 0 & 0 & 0 & 0 & 0
\end{array}\right]
\eeq so that $\mbox{div}(\bs^1)=[-e_1-e_6, e_1-e_2-e_3-e_4, e_2-e_7-e_9-e_{10},e_{10}-e_{11},-e_{8}+e_9+e_{11},e_3-e_5+e_7+e_8,e_4+e_5+e_6]^T$.\\
If we consider equation (\ref{s_decomp}) in the case $k=1$, \begin{equation}
\label{s_decomp_k=1}
   \bs^1=\mathbf{B}_{1}^T\, \bs^{0}+\bs_{H}^1 +\mathbf{B}_{2}\, \bs^{2},
\end{equation}
recalling that $\mB_1 \mB_2=\mathbf{0}$, it is easy to check that the first component in (\ref{s_decomp_k=1}) has zero curl, and then it may be called an {\it irrotational} component, whereas the third component has zero divergence, and then it may be called a {\it solenoidal} component, in analogy to the calculus terminology used for vector fields. The {\it harmonic} component $\mathbf{s}_{H}^1$ is a flow vector that is both curl-free and divergence-free.
Notice also that, in \eqref{s_decomp_k=1}, $\mathbf{B}_{1}^T\, \bs^{0}$ represents the (discrete) {\it gradient} of $\bs^0$.
\subsection{Topology-aware unitary basis}
In this section, we propose a  method to build a unitary basis to represent edge signals, reflecting  topological properties of the complex, more specifically triple-wise relations. The idea underlying the method is to search a basis that minimizes a
real-valued function capturing triple-wise relations.  For example, in the graph case, a key topological property is connectivity, which is well captured by the {\it cut} function. The associated real-valued function can be built using the so called {\it Lov\'{a}sz extension} \cite{Bach2013}, \cite{Lovasz1983} of the cut size.
\begin{figure}[!htp]
	\centering
	\includegraphics[width=7.8cm, height =5.6cm]{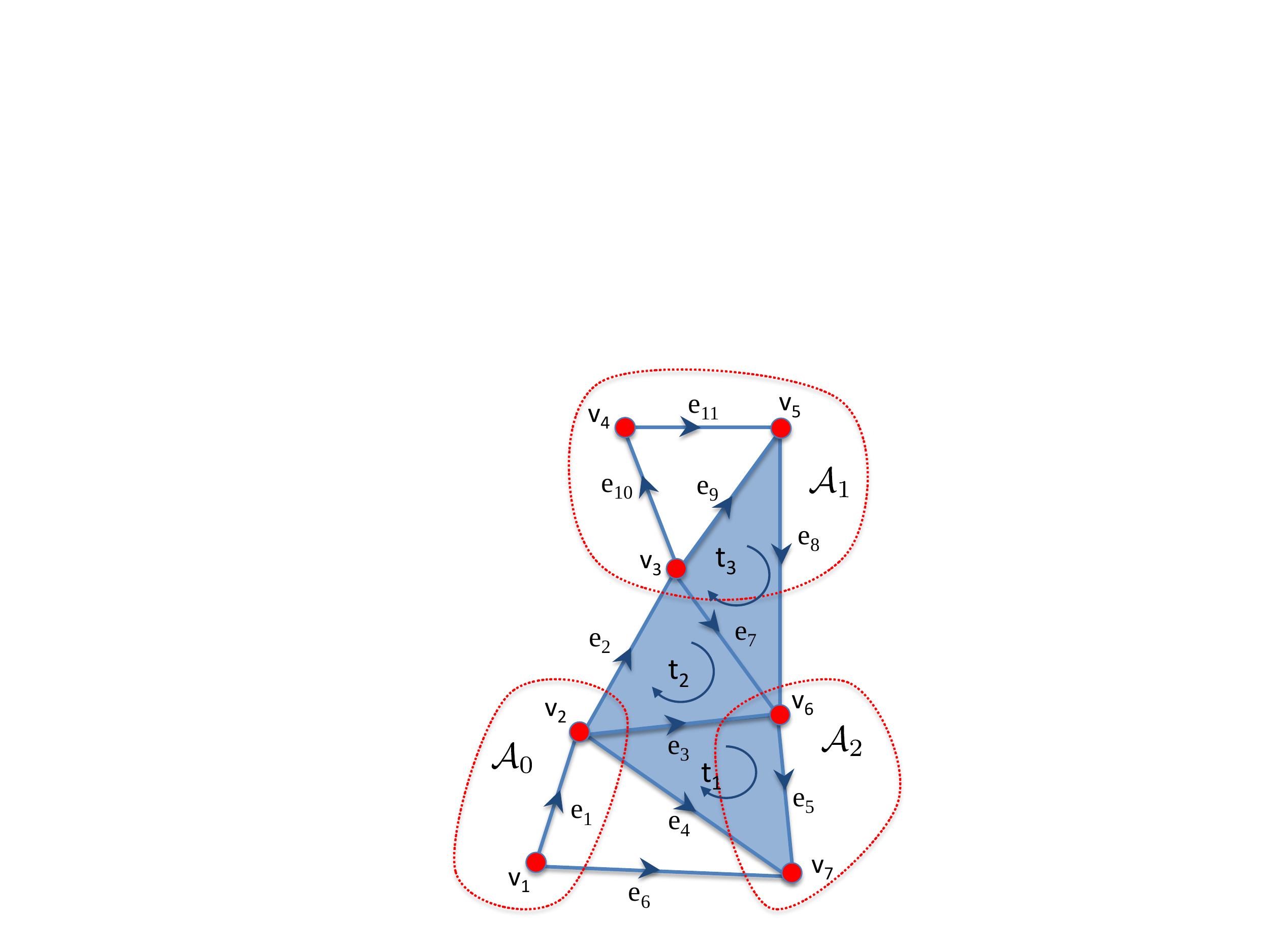}
	\caption{Cut of order $1$.}
	\label{cut}
\end{figure}
More specifically, given a graph ${\cal G({\cal V}, {\cal E})}$
and a partition of its vertex set ${\cal V}$ in two sets ${\cal A}_0$ and ${\cal A}_1$, with  ${\cal A}_0 \cup {\cal A}_1={\cal V}$ and ${\cal A}_0 \cap {\cal A}_1=\emptyset$, the cut size is defined as
\begin{equation}
F_0({\cal A}_0, {\cal A}_1)={\rm cut}({\cal A}_0, {\cal A}_1)\triangleq \sum_{i\in {\cal A}_0}\sum_{j\in {\cal A}_1} a_{ij}
\end{equation}
where $a_{ij}=1$ if $(i, j)\in {\cal{E}}$ and  $a_{ij}=0$ otherwise.
$F_0({\cal A}_0, {\cal A}_1)$ is a {\it set} function and is known to be a submodular function \cite{Bach2013}. Now, we want to translate the set function $F_0({\cal A}_0, {\cal A}_1)$ onto a real-valued function $f_0(\bx^0)$, defined over ${\mathbb R}^V$, to be used for the subsequent optimization. This can be done using the so called {\it Lov\'{a}sz extension} \cite{Bach2013},  which is equal to:
\begin{equation}
\label{l1-total variation}
    f_0(\bx^0)=\sum_{i=1}^{V}\sum_{j=1}^{V} a_{ij} |x_i^0-x_j^0|.
\end{equation}
with $\bx^0 \in \mathbb{R}^{V}$.
The function in \eqref{l1-total variation} measures the total variation of a signal defined over the nodes of a graph and then it can be used as a regularization function, whenever the signal of interest is known to be a smooth function.
The function in \eqref{l1-total variation} formed the basis of the method proposed in \cite{sardellitti2017graph} to build a unitary basis for analyzing signals defined over the vertices of a graph. More specifically, in \cite{sardellitti2017graph} the basis was built by solving the following optimization problem
\begin{equation}
\label{opt_U0}
\begin{array}{lll}
\mathbf{U}\triangleq \left(\mathbf{u}_1, \ldots, \mathbf{u}_{V}\right)= & \!\!\!  \underset{\mathbf{U} \in \mathbb{R}^{V \times V}}{\arg \min} \; \; \ds \sum_{n=2}^{V} f_0(\mathbf{u}_n)\medskip \\
&\; \,\text{s.t.} \qquad  \mathbf{U}^T\mathbf{U}=\mathbf{I}, \,\,\, \mathbf{u}_1=\frac{1}{\sqrt{V}} \mathbf{1}.
\end{array}
\end{equation}
In the above problem, the objective function to be minimized is convex, but the problem is non-convex because of the unitarity constraint.
To simplify the search for the basis, the objective function in \eqref{opt_U0} can be relaxed to become
\begin{equation}
\label{relaxed f_0}
    f_0^r(\bx^0)=\sum_{i=1}^{V}\sum_{j=1}^{V}  a_{ij} (x_i^0-x_j^0)^2.
\end{equation}
Substituting (\ref{relaxed f_0}) in (\ref{opt_U0}), we still have a non-convex problem. However, its solution is known to be given by the eigenvectors of $\mathbf{L}_0$. From this perspective, the basis typically used in the GFT, built as the eigenvectors of the Laplacian matrix, can be interpreted as the solution of the above optimization problem, after relaxation.\\
Now, we extend the previous approach to find a regularization function capturing triple-wise relations, useful to analyze signals defined over the {\it edges} of a graph. In the previous case, the analysis of {\it node} signals started from the {\it bi-partition} of a graph and then the construction of a real-valued extension of the cut size. Here, to analyze {\it edge} signals, we need to start from the {\it tri-partition} of the discrete set and define a set function counting the number of {\it triangles} gluing the three sets together. The function to be minimized should then be the real valued (Lov\'{a}sz) extension of such a set function.

The combinatorial study of simplicial complexes has attracted increasing attention and the generalization of Cheeger inequalities to simplicial complexes has been the subject of several works   \cite{Tessler}, \cite{Steenbergen}. In particular,  Hein \textit{et al.} introduced the total variation on hypergraphs as the Lov\'{a}sz extension of the hypergraph cut, defined as the size of the hyperedges set connecting a bipartition of the vertex set \cite{Hein13}.
In this work, we derive a Lov\'{a}sz extension {\it defined over the edge set}.
More formally,
we study a simplicial complex ${\cal X}({\cal V}, {\cal E}, {\cal T})$ of order two, as in the example sketched in Fig.\ref{cut}, and consider the partition of ${\cal V}$ in {\it three} sets $({\cal A}_0, {\cal A}_1, {\cal A}_2)$.
The {\it triple-wise coupling} function is now defined as
\begin{equation} \label{cut_trian}
F_1({\cal A}_0, {\cal A}_1, {\cal A}_2)=\sum_{i\in {\cal A}_0}\sum_{j\in {\cal A}_1}\sum_{i\in {\cal A}_2} a_{ijk}
\end{equation}
where  $a_{ijk}=1$ if $(i, j, k)\in {\cal{T}}$ and  $a_{ijk}=0$ otherwise.
Our main result, stated in the following theorem, is that the
Lov\'{a}sz extension of the triple-wise coupling function gives rise to a measure of the curl of edge signals along triangles.
\begin{theorem}
\textit{Let $\mathcal{A}_0,\mathcal{A}_1,\mathcal{A}_2$ be a partition of the vertex set $\mathcal{V}$ of the $2$-dimensional simplicial complex
$\mathcal{X}(\mathcal{V},\mathcal{E},\mathcal{T})$ with $|\mathcal{E}|=E$, $|\mathcal{T}|=T$.
 Then the Lov\'{a}sz  extension $f : \mathbb{R}^{E} \rightarrow \mathbb{R}$, evaluated at $\bx^1 \in \mathbb{R}^{E}$, of the  triple-wise coupling size $F_1(\mathcal{A}_0,\mathcal{A}_1,\mathcal{A}_2)$ defined in (\ref{cut_trian}), is
\beq \label{eq:lovasz}
f_1(\bx^1)= \sum_{i,j,k=1}^{E} a_{ijk} |x^1_i-x^1_j+x^1_k|\,= \sum_{j=1}^{T} \left|\sum_{i=1}^{E} B_2(i,j) x_i^1\right|
\eeq
where $B_2(i,j)$  are the edge-triplet incidence coefficients defined in (\ref{inc_coeff}).}
\end{theorem}
\begin{proof} Please see Appendix A.
\end{proof}\\
The function in \eqref{eq:lovasz} represents the sum of the absolute values of the curls over all the triangles. Differently from \cite{Hein13}, where the total variation over a hypergraph was built from a bipartition of a discrete set and it was a function defined over $\mathbb{R}^V$, in our case, we start from a {\it triparition}  of the original set and we build a real-valued extension, defined over $\mathbb{R}^E$, i.e. a space of dimension equal to the number of edges. A suitable unitary basis for representing (curling) edge signals can then be found as the solution of the following optimization problem
\begin{equation}
\label{opt_U1}
\begin{array}{lll}
\mathbf{U}\triangleq \left(\mathbf{u}_1, \ldots, \mathbf{u}_{E}\right) = &  \! \!\!\underset{\mathbf{U} \in\mathbb{R}^{E \times E}}{\arg \min} \; \; \ds\sum_{n=1}^{E} f_1(\mathbf{u}_n) \medskip\\
&\;\text{s.t.} \qquad \, \mathbf{U}^T\mathbf{U}=\mathbf{I}.
\end{array}
\end{equation}
The objective function is convex, but the problem is non-convex because of the unitarity constraint. The above problem can be solved resorting to the algorithm proposed in \cite{sardellitti2017graph}, tuned to the new objective function given in \eqref{eq:lovasz}. Alternatively,  similarly to what is typically done with graphs,  $f_1(\bx^1)$ can be relaxed and replaced with
\begin{equation}
\label{relaxed_TV}
    f_1^r(\bx^1)=\ds \sum_{j=1}^{T} \left(\sum_{i=1}^{E} B_2(i,j) x^1_i\right)^2= \bx^{1 \, T} \mB_2 \mB_2^T \bx^1.
\end{equation}
Substituting this function back to \eqref{opt_U1}, the solution is given by the eigenvectors  of $\mB_2 \mB_2^T$.

The above arguments show that the eigenvectors of $\mB_2 \mB_2^T$ provide a suitable basis to analyze edge signals having a curling behavior. However, as we know from the Hodge decomposition recalled in the previous section, edge signals contain other useful components that are orthogonal to solenoidal signals, namely irrotational and harmonic components. Then, in general, it is useful to take as a unitary basis for analyzing edge signals {\it all} the eigenvectors of $\mL_1$, i.e. the eigenvectors associated to the nonzero eigenvalues of $\mB_2 \mB_2^T$ and of $\mB_1^T \mB_1$, plus the eigenvectors associated to the kernel of $\mL_1$.
In summary, the theory developed in this section provides a further argument, rooted on intrinsic topological properties of the simplicial complex on which the signal is defined,  to exploit the Hodge decomposition to find a suitable basis for the analysis of edge signals. Furthermore, the theory says that using the Laplacian eigenvectors comes as a consequence of a relaxation step.
In many cases, whenever the numerical complexity is not an issue, it may be better to keep the $\ell_1$-norm objective functions given in \eqref{l1-total variation} and \eqref{eq:lovasz}, as this would yield more sparse, and then more informative, representations, as a generalization of what done for directed graphs in \cite{sardellitti2017graph}.

\section{Edge flows estimation}
\label{sec:edge_flows_est}
Let us consider now the observation of an edge signal affected by additive noise. Our goal in this section, is to formulate the denoising problem as a constrained problem, rooted on the Hodge decomposition. Denoising edge signals embedded in noise was already considered in \cite{grady2010},\cite{Segarra_018},\cite{Segarra_019} and, more specifically, in \cite{jiang2011}. The formulation proposed in \cite{jiang2011} was based on the definition of signals over simplicial complexes as skew-symmetric arrays of dimension equal to the order of the corresponding simplex plus one. So, the edge flow was represented as a matrix (i.e., a two-dimensional array) satisfying the constraint $X(i, j)=-X(j, i)$. A signal defined over the triangles was represented as an array of dimension three, satisfying the property $\Phi(i, j, k)=\Phi(j, k, i)=\Phi(k, i, j)=-\Phi(j, i, k)=-\Phi(i, k, j)=-\Phi(k, j, i)$. Our aim in this section, is to formulate the denoising optimization problem dealing only with vectors.
Based on \eqref{s_decomp}, the observed vector can be modeled as
\begin{equation}
\label{noisy_observation}
    \bx^1=\mathbf{B}_{1}^T\, \bs^{0}+\bs_{H}^1 +\mathbf{B}_{2}\, \bs^{2}+\bv^1
\end{equation}
where $\bv^1$ is noise. Let us suppose, for simplicity, that the noise vector is Gaussian, with zero-mean entries all having the same variance $\sigma_n^2$. The optimal estimator can then be formulated as the solution of the following problem
\beq \label{eq:p_tot}
\begin{array}{lll}
(\hat{\bs}^0,\hat{\bs}^2,\hat{\bs}^1_{\text{H}})= \hspace{-0.3cm}
 \underset{{{\bs}^0 \in \mathbb{R}^{V},{\bs}^2 \in \mathbb{R}^{T} ,{\bs}^1_{\text{H}}\in \mathbb{R}^{E}}}{\text{argmin}}  &\hspace{-0.52cm} \parallel \mathbf{B}_2 \bs^2+
\mathbf{B}_1^T \bs^0+\bs^1_H-\bx^1 \!\parallel^2   \\
\hspace{3cm} \text{s.t.} &  \hspace{-0.41cm} \mathbf{B}_1\bs^1_H=\mathbf{0} \medskip \\
&  \hspace{-0.41cm}  \mathbf{B}_2^T\bs^1_H=\mathbf{0}\qquad (\mathcal{Q}).\\
 \end{array}
\eeq
Note that problem  $\mathcal{Q}$ is convex. Then, there exists multipliers $\boldsymbol{\lambda}_1 \in
\mathbb{R}^{V},\boldsymbol{\lambda}_2 \in
\mathbb{R}^{T}$ such that the tuple
$(\hat{\bs}^0,\hat{\bs}^2,\hat{\bs}^1_{\text{H}},\boldsymbol{\lambda}_1, \boldsymbol{\lambda}_2)$ satisfies the Karush-Kuhn-Tucker (KKT) conditions of
$\mathcal{Q}$ (note that Slater's constraint qualification is satisfied \cite{boyd2004convex}).
The associated Lagrangian function is
\beq
\begin{split}
\mathcal{L}({\bs}^0,{\bs}^2,{\bs}^1_{\text{H}},\boldsymbol{\lambda}_1, \boldsymbol{\lambda}_2)=(\mathbf{B}_2 \bs^2+
\mathbf{B}_1^T \bs^0 +\bs^1_H -\bx^1)^T \\ \cdot (\mathbf{B}_2 \bs^2+
\mathbf{B}_1^T \bs^0 +\bs^1_H -\bx^1)+\boldsymbol{\lambda}_1^T \mathbf{B}_1 \bs^1_H+
\boldsymbol{\lambda}_2^T \mathbf{B}_2^T \bs^1_H.
\end{split}
\eeq
Exploiting the orthogonality property $\mathbf{B}_1 \mathbf{B}_2=\mathbf{0}$, it is easy to get the following KKT conditions
\beq \hspace{-0.1cm}
\begin{array}{lllll}
\text{(a)} & \!\!\! \nabla_{\bs^0}\mathcal{L}({\bs}^0,{\bs}^2,{\bs}^1_{\text{H}},\boldsymbol{\lambda}_1, \boldsymbol{\lambda}_2)= \mathbf{B}_1 \mathbf{B}_1^T \bs^0- \mathbf{B}_1 \bx^1=\mathbf{0} \medskip\\
\text{(b)} & \!\!\! \nabla_{\bs^2}\mathcal{L}({\bs}^0,{\bs}^2,{\bs}^1_{\text{H}},\boldsymbol{\lambda}_1, \boldsymbol{\lambda}_2)= \mathbf{B}_2^T \mathbf{B}_2 \bs^2- \mathbf{B}_2^T \bx^1=\mathbf{0}\medskip\\
\text{(c)} & \!\!\! \nabla_{\bs^1_{H}}\mathcal{L}({\bs}^0,{\bs}^2,{\bs}^1_{\text{H}},\boldsymbol{\lambda}_1, \boldsymbol{\lambda}_2)\!=\!
\bs^{1}_{H}-\bx^{1}\!+\mathbf{B}_1^T \boldsymbol{\lambda}_1\!+\mathbf{B}_2 \boldsymbol{\lambda}_2\!=\!\mathbf{0}\medskip\\
\text{(d)} &  \!\!\! \mathbf{B_1}\bs^1_H=\mathbf{0}, \;\, \mathbf{B}_2^T\bs^1_H=\mathbf{0}\medskip\\
\text{(e)} & \!\!\! \boldsymbol{\lambda}_1 \in
\mathbb{R}^{V},\boldsymbol{\lambda}_2 \in
\mathbb{R}^{T}.
\end{array}  \nonumber
\eeq
Note that conditions (a)-(c) reduce to
\beq
\begin{split}
\label{eq:all_cnd}
\text{(a)}  & \; \mathbf{L}_0 \bs^0= \mathbf{B}_1 \bx^1 \\
\text{(b)} &  \;\mathbf{B}_2^T \mathbf{B}_2 \bs^2= \mathbf{B}_2^T \bx^1\\
\text{(c)} & \;\bs^{1}_{H}=\bx^{1}-\mathbf{B}_1^T \boldsymbol{\lambda}_1-\mathbf{B}_2 \boldsymbol{\lambda}_2.\\
\end{split}
\eeq
Multiplying both sides of condition (c) by $\mathbf{B}_2^T$, and using the second equality in (d) and condition (b), we get
\beq
\mathbf{B}_2^T \mathbf{B}_2 \bs^2= \mathbf{B}_2^T \mathbf{B}_2\boldsymbol{\lambda}_2.
\eeq
This means that $\bs^2$ and $\blambda_2$ may differ only by an additive vector lying in the nullspace of $\mathbf{B}_2^T \mathbf{B}_2$. Let us set  $\bs^2=\boldsymbol{\lambda}_2+\bc_2$, with $\mathbf{B}_2^T \mathbf{B}_2\bc_2=\b0$.
Similarly,  multiplying (c)  by $\mathbf{B}_1$ and using the first equality in (d) and condition (a), we obtain
\beq
\mathbf{B}_1 \mathbf{B}_1^T \bs^0= \mathbf{B}_1 \mathbf{B}_1^T \boldsymbol{\lambda}_1,
\eeq
which implies that $\bs^0=\boldsymbol{\lambda}_1+\bc_1$, with $\bc_1$ such that $\mathbf{B}_1 \mathbf{B}_1^T \bc_1=\b0$.
Thus, condition (c) reduces to
  \beq
 \bs^{1}_{H}=\bx^{1}-\mathbf{B}_1^T \bs^0-\mathbf{B}_2 \bs^2
 \eeq
 which says, as expected, that we can derive the harmonic component by subtracting the estimated solenoidal and irrotational parts from
 the observed flow signal $\bx^1$.
To recover the irrotational flow $\bs^1_{\text{irr}}$ from the $0$-order signal $\bs^0$ we need to solve equation (a) in  (\ref{eq:all_cnd}).
Note that $\mathbf{L}_0$ is not invertible. For connected graphs, it has rank $V-1$ and its kernel is the span of the vector  $\mathbf{1}$ of all ones. However, since the vector $\mathbf{b}=\mathbf{B}_1 \bx^{1}$ is also orthogonal to $\mathbf{1}$, the normal equation $\mathbf{L}_0 \bs^{0}=\mathbf{B}_1 \bx^{1}$ admits the  nontrivial solution (at least for connected graphs):
$$\bs^{0}=\mathbf{L}_0^{\dagger}  \mathbf{B}_1 \bx^{1}$$
where $^{\dagger}$ denoted the Moore-Penrose pseudo-inverse.

Similarly, the $2$-order signal $\bs^2$, solution of the  second equation in  (\ref{eq:all_cnd}), can be obtained as
$${\bs}^2=(\mathbf{B}_2^T \mathbf{B}_2)^{\dagger} \mathbf{B}_2^T \bx^1$$
since $\mathbf{B}_2^T \bx^1$  is orthogonal to the null space of $\mathbf{B}_2^T \mathbf{B}_2$.
The irrotational, solenoidal and harmonic components can then be recovered as follows
\beq
\begin{array}{llll}
 \hat{\bs}^1_{\text{irr}}=\mB_1^T \hat{\bs}^0=\mB_1^T \mathbf{L}_0^{\dagger}  \mathbf{B}_1 \bx^{1} \\
 \hat{\bs}^1_{\text{sol}}=\mB_2 \hat{\bs}^2=\mB_2 (\mathbf{B}_2^T \mathbf{B}_2)^{\dagger} \mathbf{B}_2^T \bx^1\\
 \hat{\bs}^1_{\text{H}}=\bx^1-\hat{\bs}^1_{\text{irr}}-\hat{\bs}^1_{\text{sol}}\medskip.
\end{array}
\eeq
Note that the first two conditions in (\ref{eq:all_cnd}) imply that the variables $\bs^0$,$\bs^2$ in $\mathcal{Q}$ are indeed decoupled so that the optimal solutions coincide with those of the following problems
\beq \label{eq:p_tot0}
\begin{array}{lll}
\hat{\bs}^0=
& \underset{{\bs}^0 \in \mathbb{R}^V}{\text{argmin}} & \parallel
\mathbf{B}_1^T \bs^0  -\bx^1\parallel^2   \qquad (\mathcal{Q}_0)\\
 \end{array}
\eeq
\beq \label{eq:p_tot1}
\begin{array}{lll}
\hat{\bs}^2=
& \underset{{\bs}^2 \in \mathbb{R}^T}{\text{argmin}} & \parallel
\mathbf{B}_2 \bs^2  -\bx^1\parallel^2   \qquad (\mathcal{Q}_2).\\
 \end{array}
\eeq
\begin{figure}[!htp]
\centering
\includegraphics[width=0.35\textwidth]{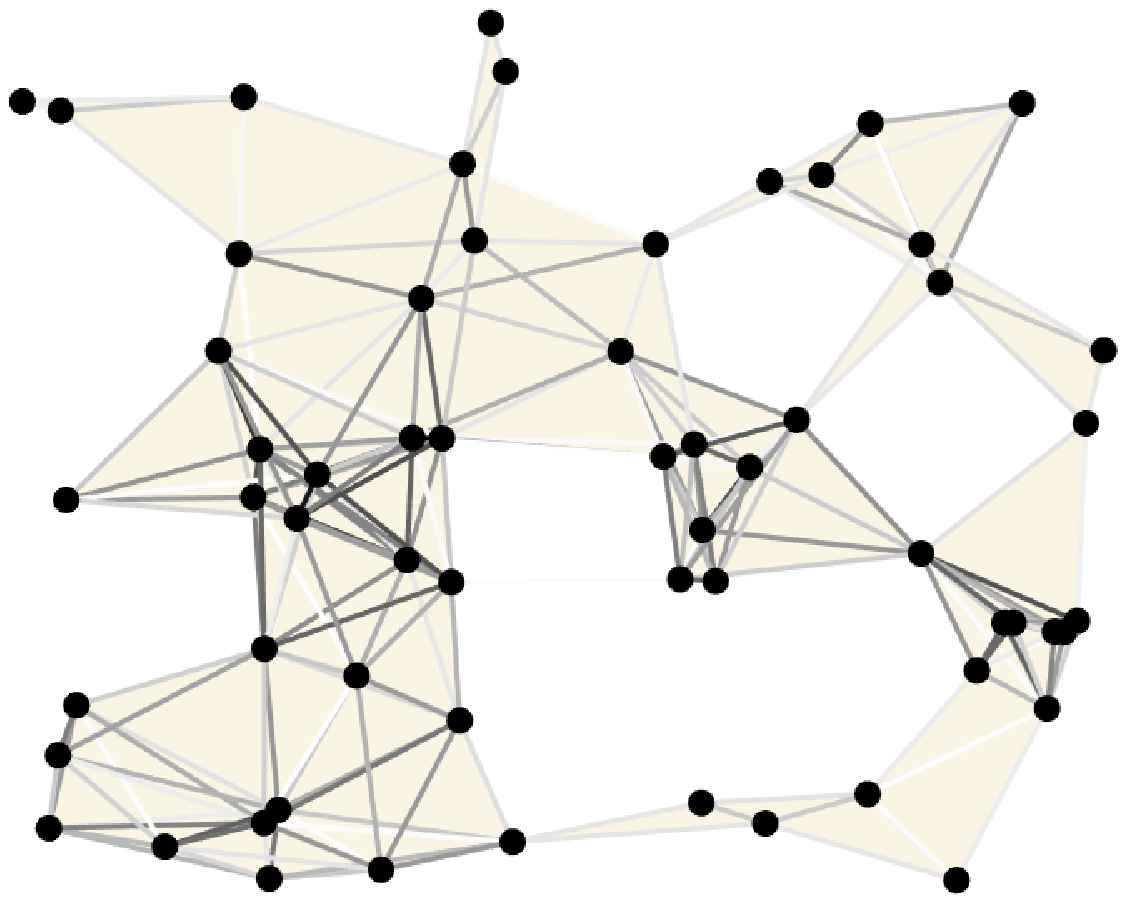}\\ (a)\\
\includegraphics[width=0.35\textwidth]{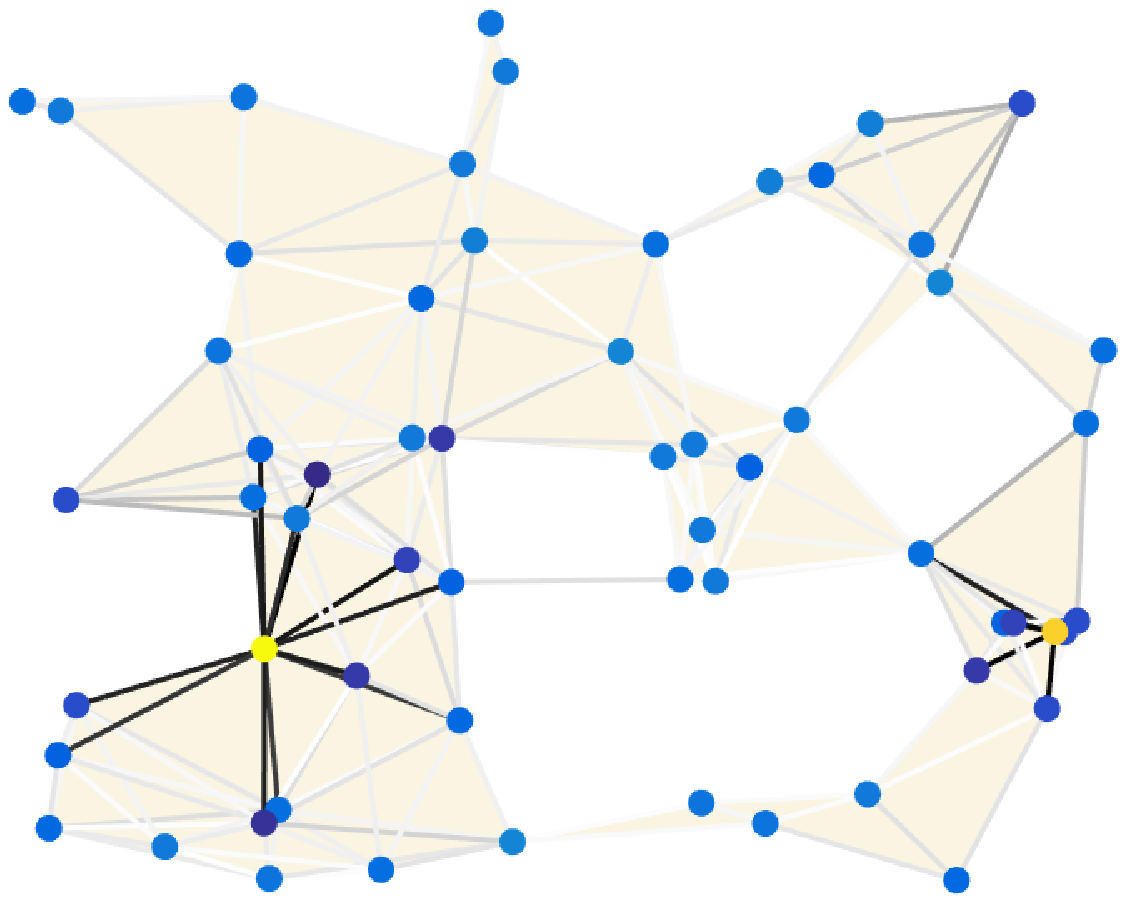}\\ (b)\\
\caption{Observed flow on a simplicial complex (a) and reconstruction of the irrotational flow (b).}
\label{edge-filtering}
\end{figure}

An example of an edge (flow) signal is reported in Fig. \ref{edge-filtering} (a), representing the simulation of data packet flow over a network, including measurement errors. The signal values are encoded in the gray color associated to each link. Suppose now, that the goal of processing is to detect nodes injecting anomalous traffic in the network, starting from the traffic shown in Fig. \ref{edge-filtering}(a). If some node is a source of an anomalous traffic, that node generates an edge signal with a strong irrotational (or divergence-like) component.
To detect spamming nodes, we can then project the overall observed traffic vector  $\bs^1$ onto the space orthogonal to the space spanned  by the nonzero eigenvectors of $\mB_2\mB_2^T$, computing
\begin{equation}
\by^1=\left(\mI-\mU_{\text{sol}}\mU_{\text{sol}}^T\right) \bs^{1}
\end{equation}
where $\mU_{\text{sol}}$ is the matrix whose columns collect the eigenvectors associated with the nonzero eigenvalues of $\mB_2\mB_2^T$. The result of this projection is reported in Fig. \ref{edge-filtering}(b), where we can clearly see that the only edges with a significant contributions are the ones located around two nodes, i.e. the source nodes injecting traffic in the network,  whose divergence is encoded by the node color.

\section{Sampling and recovering of signal defined over simplicial complex}
\label{sec:sampling}
Suppose now that we only observe a few samples of a $k$-order signal. The question we address here is to find the conditions to recover the whole signal $\bs^k$ from a subset of samples. To answer this question, we may use the theory developed in \cite{tsitsvero2016signals} for signals on graph, and later extended to hypergraphs in \cite{barbarossa2016introduction}.
For simplicity, we focus on  signals defined over a simplicial complex of order $K=2$, i.e. on vertices, edges and triangles.
Given a set of edges $\mathcal{S}\subseteq \mathcal{E}$
  we define an edge-limiting operator as a diagonal matrix   $\mathbf{D}_{\mathcal{S}}$ of dimension equal to the number of edges, with a one in the positions where we measure the flow, and zero elsewhere, i.e.
  \beq
  \mathbf{D}_{\mathcal{S}}=\text{diag}\{\mathbf{1}_{\mathcal{S}}\}
  \eeq
where $\mathbf{1}_{\mathcal{S}}$ is the set indicator vector whose $i$-th entry is equal to one if the edge $e_i \in \mathcal{S}$
and zero otherwise.
We say that an edge signal $\bs^1$ is perfectly localized over the subset $\mathcal{S}\subseteq \mathcal{E}$ (or $\mathcal{S}$-edge-limited) if
 $\bs^1 = \mathbf{D}_{\mathcal{S}}\bs^1$.
Similarly, given the matrix $\mathbf{U}_1$ whose columns are the eigenvectors of $\mathbf{L}_1$, and a  subset of indices $\mathcal{F}$,  we define the operator
 \beq
 \mathbf{F}_{\mathcal{F}}=\mathbf{U}_1 \mathbf{\Sigma}_\mathcal{F}\mathbf{U}^T_1
 \eeq
 where  $\mathbf{\Sigma}_\mathcal{F}=\text{diag}(\mathbf{1}_{\mathcal{F}})$.
An edge signal $\bs^1$ is {\it $\mathcal{F}$-bandlimited} over a frequency set $\mathcal{F}$
if  $\mathbf{F}_{\mathcal{F}} \bs^1=\bs^1$.
The operators $\mathbf{D}_{\mathcal{S}}$ and $\mathbf{F}_{\mathcal{F}}$ are self-adjoint and idempotent and represents orthogonal projectors, respectively, on the sets ${\mathcal{S}}$ and ${\mathcal{F}}$.
If we look for edges signals which are perfectly localized in both the edge and frequency domains, some conditions for perfect localization  have been derived in \cite{tsitsvero2016signals}. More specifically: a)  $\bs^1$ is perfectly localized over both  the edge set $\mathcal{S}$ and  the frequency set $\mathcal{F}$ if and only if  the operator
$\mathbf{F}_{\mathcal{F}} \mathbf{D}_{\mathcal{S}} \mathbf{F}_{\mathcal{F}}$ has an eigenvalue equal to one, i.e.
\beq \label{eq:loc_both}
\| \mathbf{D}_{\mathcal{S}}\mathbf{F}_{\mathcal{F}}\|_2=\|\mathbf{F}_{\mathcal{F}} \mathbf{D}_{\mathcal{S}}\|_2=\|\mathbf{F}_{\mathcal{F}} \mathbf{D}_{\mathcal{S}} \mathbf{F}_{\mathcal{F}}\|_2=1\eeq where $\|\bA\|_2$ denotes the spectral norm of $\bA$; b) a sufficient condition for the existence of such signals is that $|\mathcal{S}|+|\mathcal{F}|>E$. Conversely, if $|\mathcal{S}|+|\mathcal{F}|\leq E$, there could still exist perfectly localized vectors, when condition (\ref{eq:loc_both}) holds.

In the following we first consider sampling on a single layer by extracting samples  from the edge signals  and then, we consider multi-layer processing using samples of signals defined over different layers of the simplex.
 \subsection{Single-layer sampling}
 In the following theorem  \cite{tsitsvero2016signals} we provide a necessary and sufficient condition to recover the edge signal $\bs^1$ from its samples $\bs^{1}_{\mathcal{S}}\triangleq  \mathbf{D}_{\mathcal{S}} \bs^1$.
 \begin{theorem}
 \label{sampling theorem}
{\it Given the  bandlimited edge signal  $\bs^1=\mathbf{F}_{\mathcal{F}} \bs^{1}$,
 it is possible to recover $\bs^1$ from a subset of samples collected over the subset $\mathcal{S}\subseteq \mathcal{E}$ if and only if the following condition holds:
\beq \label{eq_spect_norm}
\|\bar{\mathbf{D}}_{\mathcal{S}}\mathbf{F}_{\mathcal{F}}\|_2=\|\mathbf{F}_{\mathcal{F}} \bar{\mathbf{D}}_{\mathcal{S}} \|_2<1
\eeq
with $\bar{\mathbf{D}}_{\mathcal{S}}=\mathbf{I}-\mathbf{D}_{\mathcal{S}}$.}
\end{theorem}
\begin{proof}
The proof is a straightforward extension of Th. 4.1 in  \cite{tsitsvero2016signals} to signals defined on the edges of the complex.
\end{proof}

In words, the above conditions mean that there can be no $\mathcal{F}$-bandlimited signals that are perfectly localized on the complementary set $\bar{\mathcal{S}}$. Perfect recovery of the signal $\bs^{1}$ from $\bs^{1}_{\mathcal{S}}$ can be achieved as
\begin{equation}
    \mathbf{r}^1=\mathbf{Q}_{\mathcal{S}}\bs^{1}_{\mathcal{S}}
\end{equation}
where  $\mathbf{Q}_{\mathcal{S}}=(\mathbf{I}- \bar{\mathbf{D}}_{\mathcal{S}}\mathbf{F}_{\mathcal{F}})^{-1}$. The existence of the above inverse is ensured by condition (\ref{eq_spect_norm}). In fact, the reconstruction error can be written as  \cite{tsitsvero2016signals}
\beq
\bs^1-\mathbf{Q}_{\mathcal{S}} \bs^1_{\mathcal{S}}=\bs^1-\mathbf{Q}_{\mathcal{S}} (\mathbf{I}- \bar{\mathbf{D}}_{\mathcal{S}})\bs^1=\bs^1-\mathbf{Q}_{\mathcal{S}} (\mathbf{I}- \bar{\mathbf{D}}_{\mathcal{S}} \mathbf{F}_{\mathcal{F}})\bs^1=\mathbf{0}, \nonumber
\eeq
 where we exploited in the second equality the bandlimited condition $\bs^1=\mathbf{F}_{\mathcal{F}} \bs^1$.

To make (\ref{eq_spect_norm})  holds true, we must guarantee that  $\mathbf{D}_{\mathcal{S}}\mathbf{F}_{\mathcal{F}} \bs^1\neq \b0$ or, equivalently, that the matrix $\mathbf{D}_{\mathcal{S}}\mathbf{F}_{\mathcal{F}}$  is full column rank, i.e. $\text{rank}(\mathbf{D}_{\mathcal{S}} \mathbf{F}_{\mathcal{F}})=|\mathcal{F}|$. Then, a necessary condition to ensure this holds is  $|\mathcal{S}|\geq |\mathcal{F}|$.\\
An alternative way to retrieve the overall signal $\bs^1$ from its samples can be obtained as follows.
If (\ref{eq_spect_norm}) holds true, the entire signal $\bs^1$ can be recovered from $\bs^1_{\mathcal{S}}$ as follows
\begin{equation} \label{eq:recov_s1}
\bs^1=\mathbf{U}_{\mathcal{F}}\left(\mathbf{U}_{\mathcal{F}}^T \mathbf{D}_{\mathcal{S}} \mathbf{U}_{\mathcal{F}}\right)^{-1}\mathbf{U}_{\mathcal{F}}^T \,\bs^1_{\mathcal{S}}
\end{equation}
where $\mathbf{U}_{\mathcal{F}}$ is the $E \times |\mathcal{F}|$ matrix whose columns are the eigenvectors of $\mL_1$ associated with the frequency set
$\mathcal{F}$.

\noindent{\bf Remark}: It is worth to notice that, because of the Hodge decomposition (\ref{eq:Hodge}), an edge signal always contains three components that are typically band-limited, as they reside on a subspace of dimension smaller than $E$. This means that, if one knows {\it a priori}, that the edge signal contains only one component, e.g. solenoidal, irrotational, or harmonic, then it is possible to observe the edge signal and to recover the desired component over all the edges, under the conditions established by Theorem \ref{sampling theorem}.

\vspace{-0.3cm}
\subsection{Multi-layer sampling}
In this section,  we consider the case where we take samples of signals of {\it different order} and we propose two alternative strategies to retrieve an edge signal $\bs^1$
from these samples.\\
The first approach aims at  recovering
$\bs^1$
by using both,  the  vertex  signal samples  $\bs^{0}_{\mathcal{A}}=\mathbf{D}_{\mathcal{A}} \bs^0$, with $\mathcal{A}\subseteq \mathcal{V}$, and the edge  samples $\bs^{1}_{\mathcal{S}}=\mathbf{D}_{\mathcal{S}} \bs^1$. Hereinafter, we denote by  $\mathcal{F}_{\text{irr}}$, $\mathcal{F}_{\text{sol}}$ and  $\mathcal{F}_{\text{H}}$
  the set of frequency indexes in $\mathcal{F}$ corresponding to the eigenvectors of $\mL_1$ belonging, respectively,  to the irrotational, solenoidal and harmonic subspaces. Note that, if $\bs^1$ is $\mathcal{F}$-bandlimited then also $\bs^1_{H}$, $\bs^{1}_{\text{sol}}$ and $\bs^{1}_{\text{irr}}$ are bandlimited with bandwidth, respectively,  $|\mathcal{F}_{\text{H}}|$, $|\mathcal{F}_{\text{sol}}|$ and $|\mathcal{F}_{\text{irr}}|$.  Define
  $\mathcal{F}_{\text{sH}}$ as the set of frequency indexes such that $\mathcal{F}_{\text{sH}}=\mathcal{F}_{\text{H}}\cup \mathcal{F}_{\text{sol}}$.
  Furthermore, given the matrix $\mathbf{U}_0$ with columns the eigenvectors of $\mL_0$, we define the operator $\mathbf{F}_{\mathcal{F}_0}^0=\mathbf{U}_{0} \mathbf{\Sigma}_{\mathcal{F}_0}\mathbf{U}^{T}_0$ where $|\mathcal{F}_0|$ denotes the bandwidth of $\bs^0$.
 Let $\mathcal{C}_1$ be the set of indexes in $\mathcal{F}_0$ associated with the eigenvectors belonging to $\text{ker}(\mL_0)$ and  denote by $\mU^0_{\mathcal{F}_{0}-\mathcal{C}_{1}}$ the matrix whose columns are the eigenvectors of $\mL_0$ associated with the frequency set $\mathcal{F}_{0}-\mathcal{C}_{1}$. Then, we can state the following theorem.
\begin{theorem}{\it
   Consider the second-order simplex  $\mathcal{X}(\mathcal{V},\mathcal{E},\mathcal{T})$ and the  edge signal   $\bs^1=\bs^1_{\text{sol}}+ \bs^1_{{H}}+\mB_1^T\bs^0$. Then, assume that: i)  the  vertex-signal  $\bs^0$ and the edge signal $\bs^1$ are bandlimited with bandwidth, respectively, $|\mathcal{F}_0|$ and $|\mathcal{F}|=|\mathcal{F}_{\text{sH}}|+ |\mathcal{F}_0|-c_1$, where $|\mathcal{F}_{\text{sH}}|=|\mathcal{F}_{\text{sol}}|+|\mathcal{F}_{\text{H}}|$ and
 $c_1\geq 0$ denotes the number of eigenvectors in the bandwidth of $\bs^0$ belonging to $\text{ker}(\mL_0)$; ii) the conditions
  $\|\bar{\mathbf{D}}_{\mathcal{A}}\mF_{\mathcal{F}_0}^0\|_2<1$ and $\|\bar{\mathbf{D}}_{\mathcal{S}}\mF_{\mathcal{F}_{\text{sH}}}\|_2<1$ hold true.
   Then,  it follows that:
   \begin{itemize} \item[a)] $\bs^1$ can be perfectly recovered  from both a set of  vertex  signal samples  $\bs^{0}_{\mathcal{A}}=\mathbf{D}_{\mathcal{A}} \bs^0$ and from the edge  samples $\bs^{1}_{\mathcal{S}}=\mathbf{D}_{\mathcal{S}} \bs^1$  via the set of equations
\beq
\left[\begin{array}{lll}
\bs^{0} \medskip\\
\bar{\bs}^{1}
\end{array}\right]=\mQ \left[\begin{array}{lll}
\bs^{0}_{\mathcal{A}} \medskip\\
\bs^{1}_{\mathcal{S}}
\end{array}\right],
\eeq
where $\bar{\bs}^{1}=\bs^{1}_{\text{sol}}+\bs^{1}_{{H}}$,
\beq
\begin{array}{lll}
\mQ=\left[ \begin{array}{lll}(\mathbf{I}-\overline{\mathbf{D}}_{\mathcal{A}} \mathbf{F}_{\mathcal{F}_0}^0)^{-1}  & \mathbf{O}\medskip\\
 \mP &
 (\mathbf{I}-\overline{\mathbf{D}}_{\mathcal{S}} \mathbf{F}_{\mathcal{F}_{\text{sH}}} )^{-1} \end{array}\right]\end{array}
\eeq
and $\mP=- (\mathbf{I}-\overline{\mathbf{D}}_{\mathcal{S}} \mathbf{F}_{\mathcal{F}_{\text{sH}}})^{-1}{\mathbf{D}}_{\mathcal{S}} \mB_1^T \mathbf{F}_{\mathcal{F}_0}^0 (\mathbf{I}-\overline{\mathbf{D}}_{\mathcal{A}} \mathbf{F}_{\mathcal{F}_0}^0)^{-1}$;
\item[b)]  There exists  a set of frequency indexes  $\mathcal{F}_{\text{irr}} \subset \mathcal{F}$, for which the  eigenvectors matrix $\mU_{\mathcal{F}_{\text{irr}}}$ of $\mL_1$ satisfies the equality
  $\mU_{\mathcal{F}_{\text{irr}}}= \mB^T_1 \mU^0_{\mathcal{F}_{0}-\mathcal{C}_{1}}$, and such that any $\mathcal{F}$-bandlimited edge signal with set of frequency indexes $\mathcal{F}=\mathcal{F}_{\text{irr}}\cup \mathcal{F}_{\text{sH}}$ and $|\mathcal{F}_{\text{irr}}|=|\mathcal{F}_{0
    }|-c_1$  can be recovered by using $N_0\geq |\mathcal{F}_0|$ samples from $\bs^0$ and $N_1\geq |\mathcal{F}_{\text{sH}}|$ samples from $\bs^1$.
    \end{itemize}}
\end{theorem}
\begin{proof} See Appendix B. 
\end{proof}

Let us now consider the case where we want to recover $\bs^1$ from samples collected  from signals of order $0$, $1$, and $2$, i.e. $\bs^0$, $\bs^1$ and $\bs^2$.
We denote by $\bs^2_{\mathcal{M}}=\mD_{\mathcal{M}} \bs^2$ the vector of triangle signal samples  with ${\mathcal{M}} \subseteq \mathcal{T}$.
 Furthermore, given the matrix $\mU_2$ with columns the eigenvectors of the second-order Laplacian $\mL_2$, we define the operator $\mF_{\mathcal{F}_2}^2=\mU_{2} \mathbf{\Sigma}_{\mathcal{F}_2}\mU_{2}^T$ where $|\mathcal{F}_2|$ denotes the bandwidth of $\bs^2$. Then, assuming $\bs^2$ bandlimited, it holds $\bs^2=\mF_{\mathcal{F}_2}^2 \bs^2$. Denote with $\mathcal{C}_2$ the set of indexes in $\mathcal{F}_2$ associated with the eigenvectors belonging to $\text{ker}(\mL_2)$ and with $\mU^2_{\mathcal{F}_{2}-\mathcal{C}_{2}}$ the matrix whose columns are the eigenvectors of $\mL_2$ associated with the index set  $\mathcal{F}_{2}-\mathcal{C}_{2}$.
\begin{theorem}{\it
   Consider the second-order simplex $\mathcal{X}(\mathcal{V},\mathcal{E},\mathcal{T})$, the  edge signal   $\bs^1=\mB_2 \bs^2+ \bs^1_{{H}}+\mB_1^T\bs^0$ and assume that: i)  the  vertex-signal  $\bs^0$, the edge signal $\bs^1$ and the triangle signal $\bs^2$ are bandlimited with bandwidth, respectively, $|\mathcal{F}_0|$, $|\mathcal{F}_1|=|\mathcal{F}_0|+|\mathcal{F}_{\text{H}}|+|\mathcal{F}_2|-(c_1+c_2)$ and $|\mathcal{F}_2|$, with $c_1\geq 0$ and $c_2\geq 0$ the number of eigenvectors in the bandwidth of $\bs^0$ and $\bs^2$ belonging, respectively,  to   $\mbox{ker}(\mL_0)$ and $\mbox{ker}(\mL_2)$; ii) all conditions $\|\bar{\mD}_{\mathcal{A}}\mF_{\mathcal{F}_0}^0\|_2<1$, $\|\bar{\mD}_{\mathcal{S}}\mF_{\mathcal{F}_{H}}\|_2<1$ and $\|\bar{\mD}_{\mathcal{M}}\mF_{\mathcal{F}_2}^2\|_2<1$ hold true.
   Then,  it follows:
   \begin{itemize} \item[a)] $\bs^1$ can be perfectly recovered  from  a set of  vertex  signal samples  $\bs^{0}_{\mathcal{A}}=\mD_{\mathcal{A}} \bs^0$,  from the edge  samples $\bs^{1}_{\mathcal{S}}=\mD_{\mathcal{S}} \bs^1$ and from the triangle samples $\bs^{2}_{\mathcal{M}}=\mD_{\mathcal{M}} \bs^2$   as
\beq
\left[\begin{array}{lll}
\bs^{0} \medskip\\
\bs^{1}_{{H}}\medskip\\
\bs^2
\end{array}\right]=\mR \left[\begin{array}{lll}
\bs^{0}_{\mathcal{A}} \medskip\\
\bs^{1}_{\mathcal{S}}\medskip\\
\bs^{2}_{\mathcal{M}}
\end{array}\right],
\eeq
where
\beq \hspace{-0.57cm}
\begin{array}{lll}
\mR=\left[ \begin{array}{lll}(\mI-\overline{\mD}_{\mathcal{A}} \mF_{\mathcal{F}_0}^0)^{-1}  & \mathbf{O} & \mathbf{O} \medskip\\
 \mP_1 & \mathbf{P}_2
  & \mathbf{P}_3 \medskip\\ \mathbf{O} & \mathbf{O} & (\mI-\overline{\mD}_{\mathcal{M}} \mF_{\mathcal{F}_{2}}^2 )^{-1}
 \end{array}
\!\! \right]
 \end{array}
\eeq
and $\mP_1=- (\mI-\overline{\mD}_{\mathcal{S}} \mF_{\mathcal{F}_{\text{H}}})^{-1}{\mD}_{\mathcal{S}} \mB_1^T \mF_{\mathcal{F}_0}^0 (\mI-\overline{\mD}_{\mathcal{A}} \mF_{\mathcal{F}_0}^0)^{-1}$, $\mP_2= (\mI-\overline{\mD}_{\mathcal{S}} \mF_{\mathcal{F}_{\text{H}}} )^{-1}$,
$\mP_3=- (\mI-\overline{\mD}_{\mathcal{S}} \mF_{\mathcal{F}_{\text{H}}})^{-1}{\mD}_{\mathcal{S}} \mB_2 \mF_{\mathcal{F}_2}^2 (\mI-\overline{\mD}_{\mathcal{M}} \mF_{\mathcal{F}_2}^2)^{-1}$;
\item[b)]
There exist  two sets of frequency indexes  $\mathcal{F}_{\text{sol}},\mathcal{F}_{\text{irr}} \subset \mathcal{F}$, for which the  eigenvectors of $\mL_1$ stacked in the columns of the matrices $\mU_{\mathcal{F}_{\text{sol}}}$,$\mU_{\mathcal{F}_{\text{irr}}}$  satisfy, respectively, the equality
  $\mU_{\mathcal{F}_{\text{sol}}}= \mB_2 \mU^2_{\mathcal{F}_{2}-\mathcal{C}_{2}}$, and
$\mU_{\mathcal{F}_{\text{irr}}}= \mB^T_1 \mU^0_{\mathcal{F}_{0}-\mathcal{C}_{1}}$.
  Then,
   any $\mathcal{F}$-bandlimited edge signal with frequency set $\mathcal{F}=\mathcal{F}_{\text{sol}}\cup\mathcal{F}_{\text{irr}}\cup \mathcal{F}_{\text{H}}$, and $|\mathcal{F}_{\text{irr}}|=|\mathcal{F}_{0
    }|-c_1$, $|\mathcal{F}_{\text{sol}}|=|\mathcal{F}_{2}|-c_2$,
can be recovered by using $N_0\geq |\mathcal{F}_0|$ samples from $\bs^0$,  $N_1\geq |\mathcal{F}_{\text{H}}|$ samples from $\bs^1$ and $N_2\geq |\mathcal{F}_{2}|$ samples from $\bs^2$.
    \end{itemize}}
\end{theorem}
\begin{proof} See Appendix B.
\end{proof}
\section{Estimation of discrete vector fields}
\label{Estimation of discrete vector fields}
Developing tools to process signals defined over simplicial complexes is also useful to devise algorithms for processing discrete vector fields.  The use of algebraic topology, and more specifically {\it Discrete Exterior Calculus} (DEC), has been already considered for the analysis of vector fields, especially in the field of  computer graphic. Exterior Calculus (EC) is a discipline that generalizes vector calculus to smooth manifolds of arbitrary dimensions \cite{cartan1899} and DEC is a discretization of EC on simplicial complexes \cite{Hirani_thesis}, \cite{bell2012pydec}. DEC methodologies have been already proposed in \cite{fisher2007design} to produce smooth tangent vector fields in computer graphics. Smoothing tangential vector fields using techniques based on the spectral decomposition of higher order Laplacian was also advocated in \cite{brandt2017spectral}.
In this section, building on the basic ideas of  \cite{fisher2007design}, \cite{brandt2017spectral}, we propose an alternative approach to smooth discrete vector fields. More specifically, as in  \cite{fisher2007design}, \cite{brandt2017spectral}, the proposed approach is based on three main steps: i) map the vector field onto an edge signal; ii) filter the edge signal; and iii) map the edge signal back to a vector field. Steps i) and iii) are essentially the same as in \cite{brandt2017spectral}. The difference we introduce here is in step ii), i.e. in the filtering of the edge signal. Filtering edge signals has been considered before, see, e.g., \cite{Segarra_019}, but in our case we consider a different formulation and we incorporate the proper metrics in the Hodge Laplacian, dictated from the initial mapping from the vector field to the corresponding edge signal.\\ To filter vector fields, we need first to embed the vector field and the simplicial complex into a real space. Let us denote with $\mathcal{X}$ a  simplicial complex embedded in a real space $\mathbb{R}^n$ of dimension $n$. We assume that $\mathcal{X}$ is flat, i.e. all simplices are in the same affine $n$-subspace, and well-centered, which means the circumcenter of each simplex  lies in its interior \cite{Hirani_thesis}.   A discrete vector field $\vec{\bv}$ can be considered as  a map from points $\bx_i$ of $\mathcal{X}$ to $\mathbb{R}^n$.
An illustrative example is shown in Fig. \ref{fig:field_filter} (a), where the discrete vector field is represented by the set of arrows associated to the points in a regular grid in $\mathbb{R}^2$. In many applications, it is of interest to develop tools to extract relevant information from a vector field or to reduce the effect of noise.
Our proposed approach is based on the following three main steps:\\
\noindent {\it 1) Map a vector field onto a simplicial complex signal}\\
Given the set of $N$ points  $\bx_i$, we build the simplicial complex $\mathcal{X}$ using a well-centered Delaunay triangulation with  vertices  ($0$-simplicies) in the points $\bx_i$  \cite{Hirani_thesis}.  Then,
starting from the set of vectors $\vec{\bv}(\bx_i)$,
we recover a continuous vector field through interpolation based on the Whitney basis functions $\phi_{i}(\bx)$ \cite{dodziuk1976finite}
\beq \label{eq:interp_field1}
\vec{\bv}(\bx)=\sum_{i=1}^{N} \vec{\bv}\,(\bx_i) \phi_{i}(\bx)
\eeq
where $\phi_{i}(\bx)$ are affine piecewise functions.
More specifically,   $\phi_{i}$ is an hat function on vertex $v_i$, which  takes on the value one at vertex $v_i$, i.e. $\phi_{i}(\bx_i)=1$,  is zero at all other vertices, and  affine over each 2-simplex having $v_i$ as vertex.
Then,  project the vector field onto the set of $1$-simplices (edges), giving rise to a set of scalar signals defined over the edges of the complex
\beq \label{eq:flat_discrete}
\begin{split}
x_{j}^1=\left( \frac{\vec{\bv}(\sigma^0_{j_0})+\vec{\bv}(\sigma^0_{j_1})}{2}\right) \cdot \vec{\sigma}^1_{j}, \quad j=1,\ldots, E
\end{split}
\eeq
where $j_0$ and $j_1$ are the end-points of edge $j$ and $\vec{\sigma}^1_j$ stands for the vector corresponding to $\sigma^1_j=[\sigma^0_{j_0},\sigma^0_{j_1}]$, having the same direction as the orientation of $\sigma^1_j$. The vector  $\bx^1 \in \mathbb{R}^E$ of edge signals  associated with the vector field $\vec{\bv}$, is then built as  $\bx^1=[x_{1}^1, \ldots , x_{E}^1]^T$.\\
\noindent {\it 2) Filter the signal defined on the simplicial complex}\\
The  filtering strategy aims to recover an edge flow vector $\bs^1$ that fits the observed vector $\bx^1$, it is smooth and sparse. We recover the vector $\bs^1$ as the solution of
 \beq \nonumber
 \underset{\bs^1 \in \mathbb{R}^E}{\min} \quad \parallel \bs^1-\bx^1 \parallel_2+\lambda \bs^{1 \,T} \mL_1 \bs^1 + \gamma \parallel \bs^1 \parallel_1  \quad \quad (\mathcal{P}_{\mathcal{F}})
 \eeq
where $\lambda,\gamma$ are positive penalty coefficients controlling, respectively, the signal smoothness and its sparsity.
 Since we chose the Whitney form as interpolation basis, for any two signals of order $p$, $\bx^p,\by^p \in \mathbb{R}^{N_p}$, the induced inner product is given by $\bx^{p T} \mM_p \by^{p}$, where the  $N_p \times N_p$-dimensional matrix $\mM_p$, incorporating the Whitney metric,    is derived  as in \cite{bell2012pydec}[Prop. 9.7]. Therefore, the  Hodge Laplacian $\mL_1$, weighted with the appropriate metric matrices $\mM_p$, with $p=0,1,2$, is a semidefinite positive matrix that can be written as
 \beq
 \mL_1= \mB_2 \mM_2 \mB_2^T+ \mM_1 \mB_1^T \mM_0^{-1} \mB_1 \mM_1
 \eeq
 while the fitting error norm becomes
 \beq
 \parallel \bs^1-\bx^1 \parallel_2=(\bs^1-\bx^1)^T \mM_1 (\bs^1-\bx^1).
 \eeq
\noindent {\it 3) Map the filtered signal back onto a discrete vector field}\\
Finally, given the discrete signal of order $1$, $s_{e_{ij}}^1$, $\forall e_{ij} \in \mathcal{E}$,
the generated piecewise linear vector field becomes
\cite{fisher2007design}
\beq \label{eq:sharp_disc_rec}
\vec{\bv}(\bx)= \sum_{e_{ij}}  s_{e_{ij}}^1 [\phi_i(\bx) \nabla \phi_j(\bx)-\phi_j(\bx) \nabla \phi_i(\bx)].
\eeq
 \begin{figure}[!htp]
	\centering
	\includegraphics[width=7.3cm, height =4.8cm]{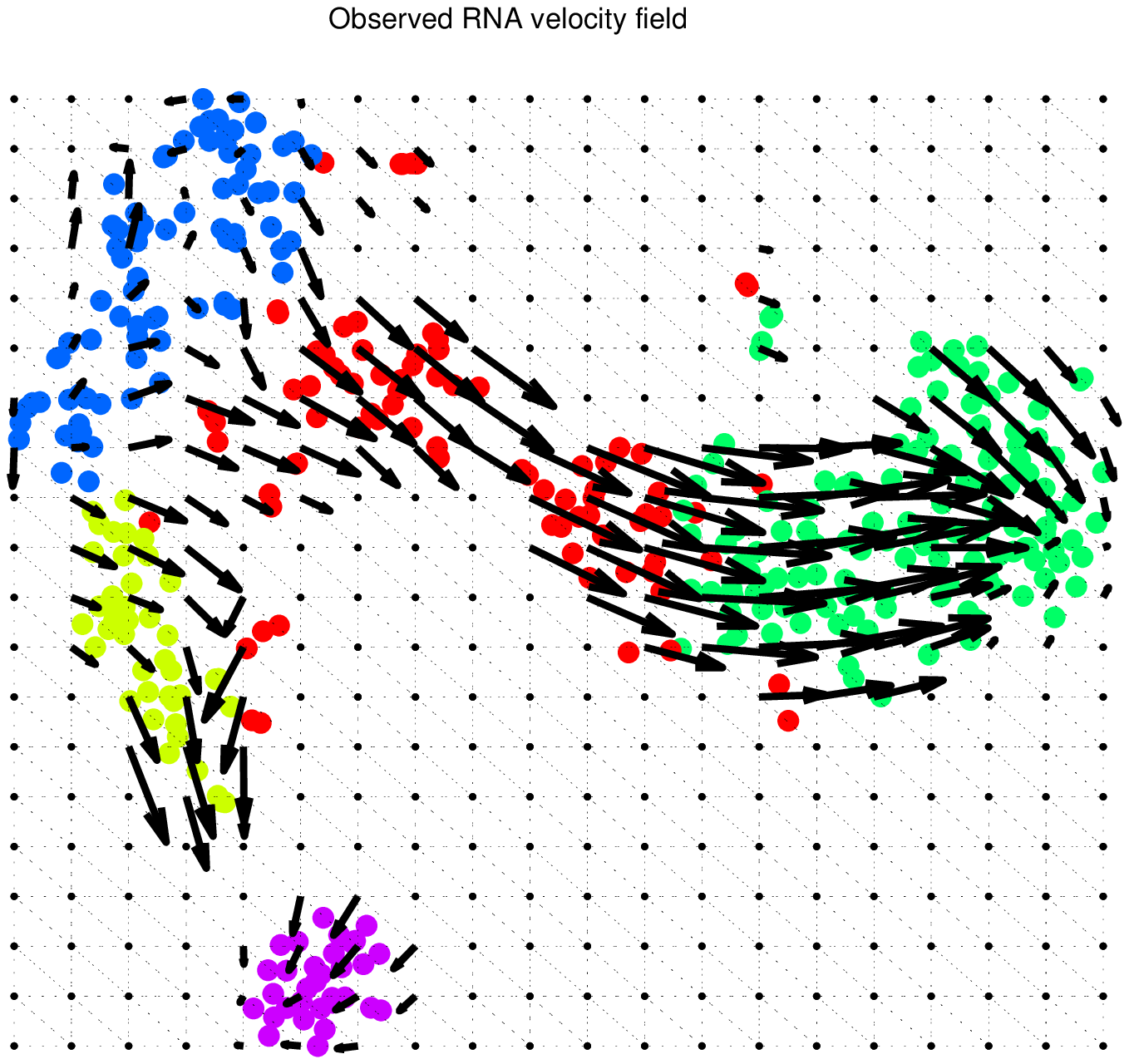}\\ (a) \vspace{0.2cm}\\
	\includegraphics[width=7.8cm, height =5cm]{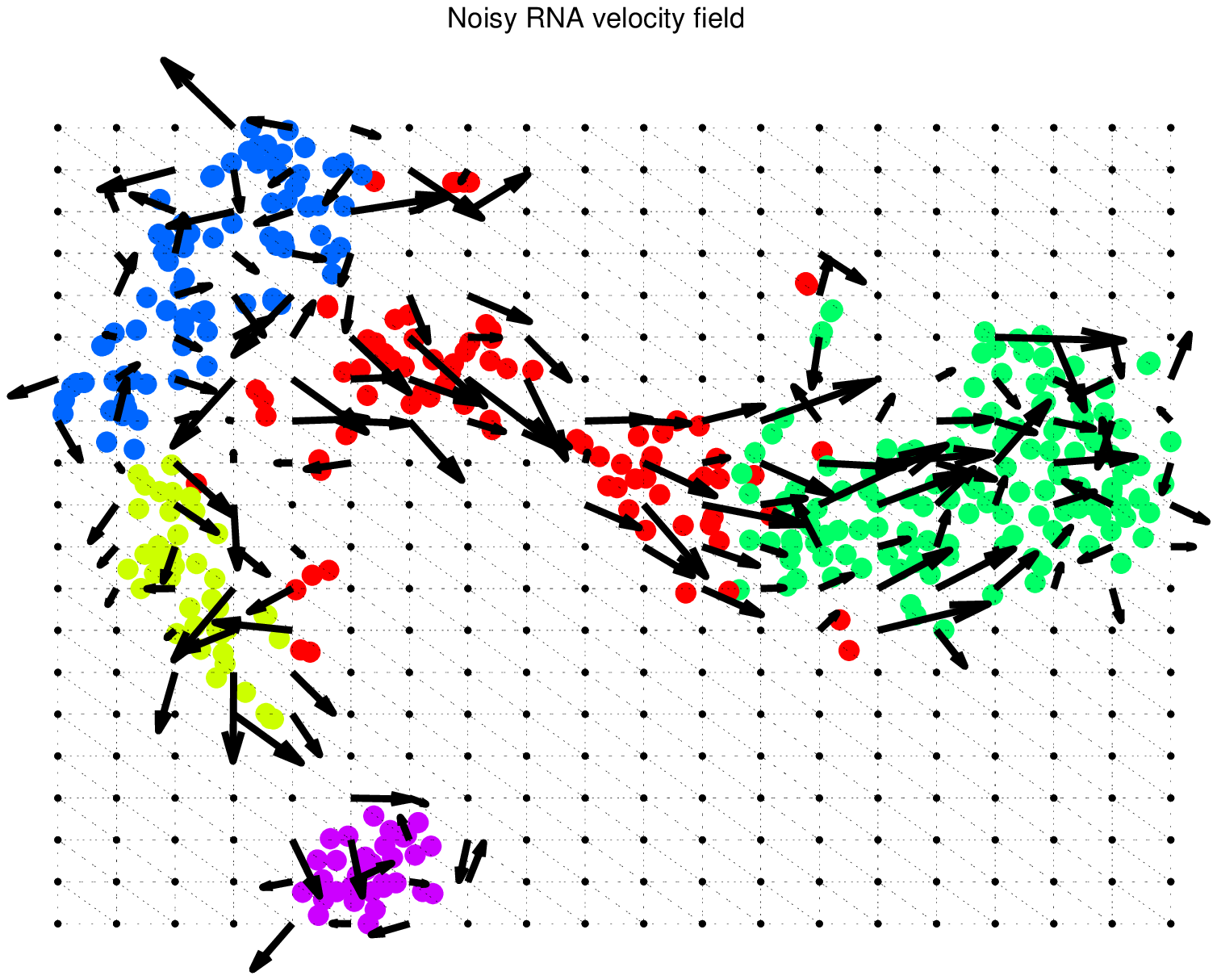}\\  (b)\vspace{0.2cm}\\
\hspace{-0.8cm}
	\includegraphics[width=8.3cm, height =5cm]{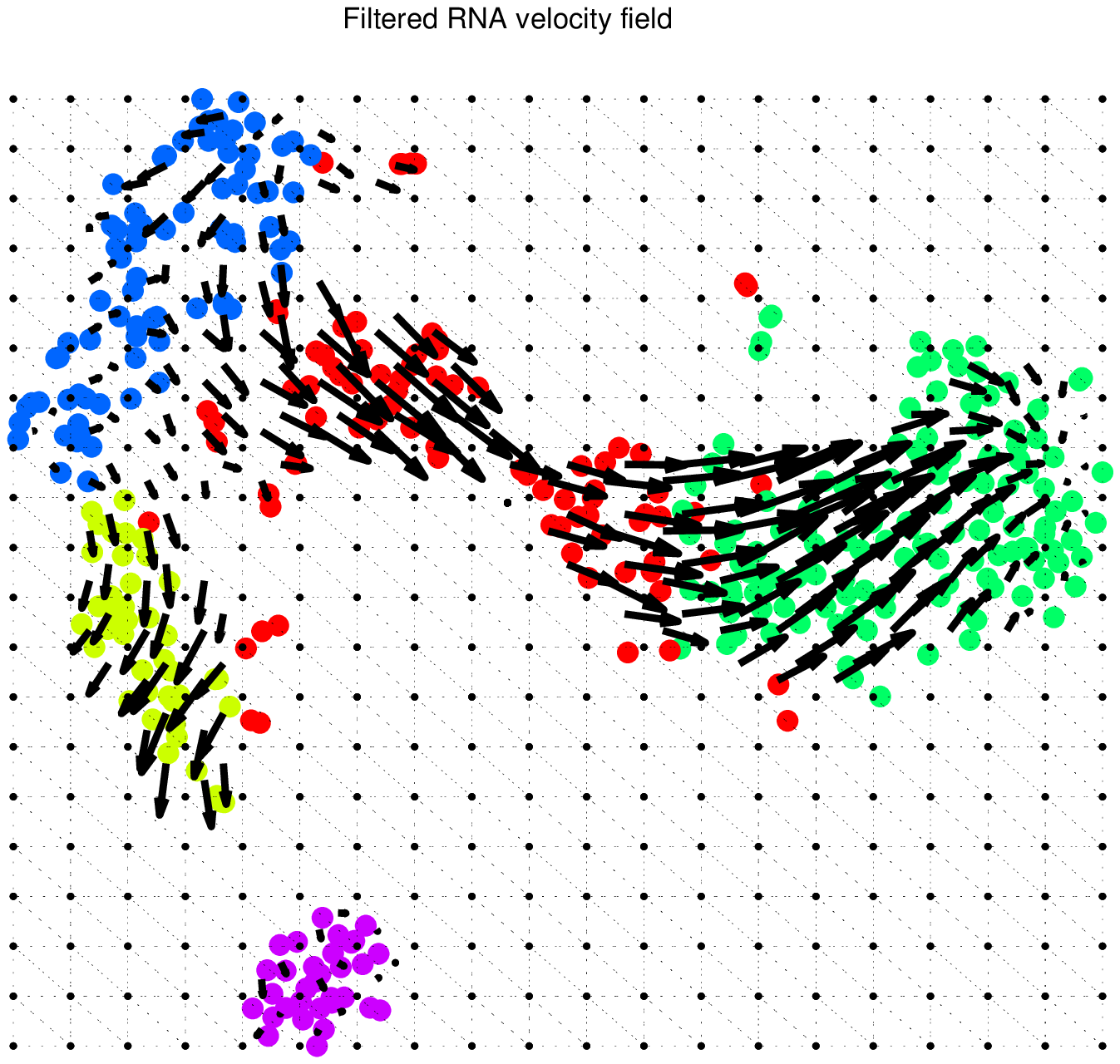}\\  (c)\vspace{0.2cm}\\
\caption{RNA velocity fields: (a) observed and (b) noisy fields; (c) reconstructed field.}
	\label{fig:field_filter}
\end{figure}
 We show now an interesting application of the above procedure to the analysis of the vector field representing the RNA velocity field, defined as the time derivative of the gene expression state \cite{la2018rna},
useful to predict the future state of individual cells and then the direction of change of the entire transcriptome during the RNA sequencing dynamic process.
The RNA velocity can be estimated by distinguishing between nascent (unspliced) and mature (spliced) messenger RNA (mRNA) in common single-cell RNA sequencing protocols \cite{la2018rna}.
We consider the mRNA-seq dataset of mouse chromaffin cell differentation analysed in  \cite{la2018rna}.
An example of RNA velocity vector field is illustrated in Fig. \ref{fig:field_filter}(a).
To analyze such a vector field using the proposed algorithms, we implemented steps 1) and 3) using the discrete exterior calculus operators provided by the PyDEC software developed in  \cite{bell2012pydec}. We considered a Delaunay well-centered triangulation of the continuous bi-dimensional space, which generates the simplicial complexes in Fig. \ref{fig:field_filter} composed of $N=400$ nodes, where we fill all the triangles. The velocity field in Fig. \ref{fig:field_filter}(a) is observed over $156$ vertices and the field vector at each vertex has been obtained  with a local Gaussian kernel smoothing \cite{la2018rna}. The underlying colored cells represented different states of the cell differentiation process. Then, to test our filtering strategy, we added a strong Gaussian noise to the observed velocity field, with  $\mbox{SNR}=0\, \mbox{dB}$, as illustrated in Fig.  \ref{fig:field_filter}(b). This noise is added to incorporate mRNA molecules degradation and model  inaccuracy.
Then we apply  the proposed filtering strategy by first  reconstructing the edge signal as in (\ref{eq:interp_field1}) and then recovering the edge signal as a solution of  the optimization problem $\mathcal{P}_{\mathcal{F}}$. Finally, we reconstruct the vector field using the interpolation formula  (\ref{eq:sharp_disc_rec}), observed at the barycentric points of each triangles.
The result is reported in Fig. \ref{fig:field_filter}(c), where we  can appreciate the robustness of the proposed filtering strategy.

\section{Inference of simplicial complex topology from data}
\label{sec:L1_inference}
The inference of the graph topology from (node) signals is a problem that has received significant attention, as shown in the recent tutorial papers \cite{giannakis2018topology}, \cite{mateos2019connecting}, \cite{dong2019learning} and in the references therein. In this section, we propose algorithms to infer the structure of a simplicial complex. Given the layer structure of a simplicial complex, we propose a {\it hierarchical} approach that infers the structure of one layer, assuming knowledge of the lower order layers.
For simplicity, we focus on the inference of a complex of order $2$ from the observation of a set of $M$ edge (flow) signals $\mX^1:=\left[\bx^1(1), \ldots, \bx^1(M)\right]$, assuming that the topology of the underlying graph is given (or it has been estimated).
So, we start from the knowledge of $\mL_0$, which implies, after selection of an orientation, knowledge of $\mB_1$. Since $\mathbf{L}_1=\mathbf{B}^T_1\mathbf{B}_1+\mathbf{B}_{2}\mathbf{B}_{2}^T$, then we need to estimate $\mB_2$. Before doing that, we check, from the data, if the term $\mB_2 \mB_2^T$ is really needed.
Since, from (\ref{s_decomp_k=1}), the only components that may depend on $\mB_2$ are the solenoidal and harmonic components, we first project the observed flow signal onto the space orthogonal to the space spanned by the irrotational component, by computing
\begin{equation}
\label{proj_on-sh}
\bx^1_{\text{sH}}(m)=\left(\mathbf{I}-\mathbf{U}_{\text{irr}}\mathbf{U}_{\text{irr}}^T\right) \bx^1(m), m=1, \ldots, M,
\end{equation}
where $\mathbf{U}_{\text{irr}}$ is the matrix whose columns are the eigenvectors associated with the  nonzero  eigenvalues of $\mB_1^T\mB_1$.
Then, denoting with   $\mX_{\text{sH}}^{1}=\left[\bx^1_{\text{sH}}(1), \ldots, \bx^1_{\text{sH}}(M)\right]$ the signal matrix of size $E \times M$, we measure the energy of $\mX_{\text{sH}}^{1}$ by taking its norm $\|\mX_{\text{sH}}^{1}\|_F$:
If the norm is smaller than a threshold $\eta$ of the  averaged energy of the observed data set, we stop, otherwise we proceed to estimate $\mB_2$.

The first step in the estimation of $\mB_2$ starts from the  detection of all cliques of three elements present in the graph. Their number is
$T={\rm trace}\left[\left(\mL_0-{\rm diag}(\mL_0 \mathbf{1})\right)^3\right]/6$. For each clique, we choose, arbitrarily, an orientation for the potential triangle filling it. The matrix $\mB_2$ can then be written as
\vspace{-0.2cm}
\begin{equation}
\label{B2}
    \mB_2=\sum_{n=1}^{T}t_n \mathbf{b}_n \mathbf{b}_n^T
\end{equation}
where $\mathbf{b}_n$ is the vector of size $E$ associated with the $n$-th clique, whose entries are all zero except the three entries associated with the three edges of the $n$-th clique. Those entries assume the value $1$ or $-1$, depending on the orientation of the triangle associated with the $n$-th clique. The coefficients $t_n$ in (\ref{B2}) are equal to one, if there is a (filled) triangle on the $n$-th clique, or zero otherwise. The goal of our inference algorithm is then to decide, starting from the data, which entries of $\bt:=[t_1, \ldots, t_T]$ are equal to one or zero. Our strategy is to make the association that enforces a small total variation of the observed flow signal on the inferred complex, using (\ref{relaxed_TV}) as a measure of total variation on flow signals. We propose two alternative algorithms: The first method infers the structure of $\mB_2$ by minimizing  the total variation of the observed data; the second method performs first a Principal Component Analysis (PCA) and then looks for the matrix $\mB_2$ and the coefficients of the expansion over the principal components that minimize the total variation plus a penalty on the model fitting error.\\

\noindent{\bf Minimum Total Variation (MTV) Algorithm}: The goal of this algorithm is to minimize the total variation over the observed data set, assuming knowledge of the number of triangles. The set of coefficients $\bt$ is found as solution of
\beq \label{eq:optm_T}
\begin{array}{lll}
\underset{\mathbf{t} \in \{0,1\}^{T}}{\min}
& q(\mathbf{t})\triangleq \ds\sum_{n=1}^{T}  t_{n}  \text{trace}\left(\mX_{\text{sH}}^{1 \; T}\mathbf{b}_{n} \mathbf{b}_{n}^T \mX_{\text{sH}}^{1}\right)  \quad (\mathcal{P}_{\text{MTV}}) \medskip\\
\quad \, \! \text{s.t.}
&  \parallel \mathbf{t} \parallel_0 = t^*, \quad  t_n \in \{0,1\}, \forall n,
\end{array}
\eeq
where $t^*$ is the number of triangles that we aim to detect. In practice, this number is not known, so it has to be found through cross-validation. Even though problem $\mathcal{P}_{\text{MTV}}$ is non-convex, it can be solved in closed form. Introducing the nonnegative coefficients
$c_n=\sum_{i=1}^{M} \bx_{\text{sH}}^{1 \; T}(i)\mathbf{b}_{n} \mathbf{b}_{n}^T \bx_{\text{sH}}^{1}(i)$, the solution can in fact be obtained by sorting the coefficients $c_n$ in increasing order and then selecting the triangles associated with the indices of
the  $t^*$ lowest coefficients $c_n$. Note that the proposed strategy infers the presence of triangles along the cliques having the minimum curl along its edges. Hence, we expect better performance when the edge signal contains only the  harmonic components, whose curls along the filled triangles is exactly null.\\

\noindent{\bf PCA-based Best Fitting with Minimum Total Variation (PCA-BFMTV)}:
To robustify the MTV algorithm in the case where the edge signal contains also a solenoidal component and is possibly corrupted by noise, we propose now the PCA-BFMTV algorithm that infers the structure of $\mB_2$  {\it and}
the edge signal that best fits the observed data set $\mX^1$, while at the same time exhibiting a small total variation over the inferred topology. The method starts performing a principal component analysis of the observed data by extracting the eigenvectors associated with the largest eigenvalues of the covariance matrix estimated from the observed data set.
More specifically, the proposed strategy is composed of two steps: 1) estimate the covariance matrix $\widehat{\mathbf{C}}_X$ from the edge signal data set $\mX^1_{\text{sH}}$ and builds the matrix $\widehat{\mathbf{U}}_{\text{sH}}$ whose columns are the eigenvectors associated with the $F$ largest eigenvalues of  $\widehat{\mathbf{C}}_X$;  2) model the observed data set as $\mX^1_{\text{sH}}= \widehat{\mathbf{U}}_{\text{sH}} \widehat{\mS}^{1}_{\text{sH}}$ and searches for the coefficient matrix $\widehat{\mS}^{1}_{\text{sH}}$ {\it and} the vector $\bt$ that solve the following problem
\beq \label{eq:optm_T_S}
\begin{array}{lll}
 \!\! \underset{\mathbf{t}\in \{0,1\}^T, \widehat{\mS}^{1}_{\text{sH}}\in \mathbb{R}^{F \times M}}{\min}
&  g(\mathbf{t},\widehat{\mS}^{1}_{\text{sH}}) + \! \gamma \! \parallel \! \mX^1_{\text{sH}}\!\! - \!\widehat{\mathbf{U}}_{\text{sH}} \widehat{\mS}^{1}_{\text{sH}}\parallel^{2}_{F}    \medskip\\
\hspace{1cm}\text{s.t.}
& \parallel \mathbf{t} \parallel_0 = t^*, \quad  t_n \in \{0,1\}, \forall n, \medskip \quad  (\mathcal{P}_{\mathcal{T S}})
\end{array}
\eeq
where $g(\mathbf{t},\widehat{\mS}^{1}_{\text{sH}})\triangleq \ds \sum_{n=1}^{T}  t_{n}\, \text{trace}(\widehat{\mS}_{\text{sH}}^{1 \; T}\widehat{\mathbf{U}}_{\text{sH}}^T\mathbf{b}_{n} \mathbf{b}_{n}^T \widehat{\mathbf{U}}_{\text{sH}} \widehat{\mS}_{\text{sH}}^{1})$
and $\gamma$ is a non-negative coefficient controlling the trade-off between the data fitting error and the signal smoothness. Although problem $\mathcal{P}_{\mathcal{T S}}$ is non-convex, it can be solved using an iterative alternating optimization algorithm
returning successive estimates of $\widehat{\mS}^{1}_{\text{sH}}$,  having fixed $\mathbf{t}$, and alternately $\mathbf{t}$,  given $\widehat{\mS}^{1}_{\text{sH}}$.
Interestingly, each step of the alternating optimization problem admits a closed form solution.
More specifically, at each iteration $k$, the coefficient matrix $\widehat{\mS}^{1}_{\text{sH}}[k]$ can be found as
\beq \label{eq:optm_Sn} \hspace{-0.17cm}
\begin{array}{lll} \nonumber
\widehat{\mS}^{1}_{\text{sH}}[k]=& \!\!\! \!\!\! \underset{\widehat{\mS}^{1}_{\text{sH}}\in \mathbb{R}^{F \times M}}{\arg \min}
& \hspace{-0.3cm}g(\mathbf{t}[k],\widehat{\mS}^1_{\text{sH}})+ \! \gamma \! \parallel \! \mX^1_{\text{sH}}\!\! - \!\widehat{\mathbf{U}}_{\text{sH}} \widehat{\mS}^{1}_{\text{sH}}\parallel^{2}_{F} \quad  (\mathcal{P}_{\mathcal{S}}^{k}).
\end{array}
\eeq
Defining $\mL_{\text{upp}}[k]:= \sum_{n=1}^{T}  t_{n}[k] \mathbf{b}_{n} \mathbf{b}_{n}^T $, problem $\mathcal{P}_{\mathcal{S}}^{k}$ admits the closed form solution
\beq
\widehat{\mS}^{1}_{\text{sH}}[k]=(\mI_F+ \gamma \, \widehat{\mathbf{U}}_{\text{sH}}^T \mL_{\text{upp}}[k] \widehat{\mathbf{U}}_{\text{sH}})^{-1} \widehat{\mathbf{U}}_{\text{sH}}^T\mX^1_{\text{sH}}.
\eeq
Then, given $\widehat{\mS}^{1}_{\text{sH}}[k]$, we can find the vector $\mathbf{t}[k+1]$ using the same method used to solve problem MTV, in (\ref{eq:optm_T}), i.e. setting  $c_n[k]:=\text{trace}(\widehat{\mS}_{\text{sH}}^{1}[k]^{T} \widehat{\mathbf{U}}_{\text{sH}}^T \mathbf{b}_{n} \mathbf{b}_{n}^T \widehat{\mathbf{U}}_{\text{sH}} \widehat{\mS}_{\text{sH}}^{1}[k])$
and taking the entries of  $\mathbf{t}_n[k+1]$  equal to $1$ for the indices corresponding to the first $t^*$ smallest  coefficients of $\{c_n[k]\}_{n=1}^{T}$, and $0$ otherwise. The iterative steps of the proposed strategy are reported in the box entitled Algorithm PCA-BFMTV.
Now we test the validity of our inference algorithms over both simulated and real data.

\noindent
 \begin{algorithm}[t]
\small

    \quad  {Set}  $\gamma>0$, $\mathbf{t}[0] \in \{0,1\}^{T}$, $\parallel \mathbf{t}[0] \parallel_0=t^*$,

     \quad $\mL_{\text{upp}}[0]=\ds\sum_{n=1}^{T}  t_{n}[0] \mathbf{b}_{n} \mathbf{b}_{n}^T $, $k=1$

    \quad  {\textbf{Repeat}}

 \quad \quad           {Set}
 $\widehat{\mS}^{1}_{\text{sH}}[k]=(\mI_F+ \gamma \, \widehat{\mathbf{U}}_{\text{sH}}^T \mL_{\text{upp}}[k-1] \widehat{\mathbf{U}}_{\text{sH}})^{-1} \widehat{\mathbf{U}}_{\text{sH}}^T\mX_{\text{sH}}^1$;
  \quad \quad

   \quad \quad          {Compute} $\mathbf{t}[k]$ {by} {sorting} {the} {coefficients}

    \quad \quad $c_n[k]=\text{trace}(\widehat{\mS}_{\text{sH}}^{1}[k]^{T} \widehat{\mathbf{U}}_{\text{sH}}^T \mathbf{b}_{n} \mathbf{b}_{n}^T \widehat{\mathbf{U}}_{\text{sH}} \widehat{\mS}_{\text{sH}}^{1}[k])$,

\quad \quad {and} {setting} {to} $1$ {the}  {entries} {of} $\bt[k]$ {corresponding} {to} {the} $t^*$

\quad \quad {smallest} {coefficients}, {and} $0$ {otherwise};

  \quad \quad  {Set}       $k=k+1$,

  \quad {\textbf{until}}   {\textbf{convergence}.}

    \caption*{ Algorithm PCA-BFMTV}
 \label{algorithm:Alg_II}
\end{algorithm}

\vspace{-0.3cm}
\label{sec:num_resul}

\noindent {\bf Performance on synthetic data: }
Some of the most critical parameters affecting the goodness of the proposed algorithms are the dimension of the subspaces associated with the solenoidal and harmonic components of the signal and the number of filled triangles in the complex.
In fact, in both MTV and PCA-BFMTV a key aspect is the detection of triangles as the cliques where the associated curl is minimum. Hence, if the signal contains only the harmonic component and there is no noise, the triangles can be identified with no error, because the harmonic component is null over the filled triangles. However, when there is a solenoidal component or noise, there might be decision errors.
\begin{figure}[h]
    \centering
        \includegraphics[width=3.2in, height=2.1in]{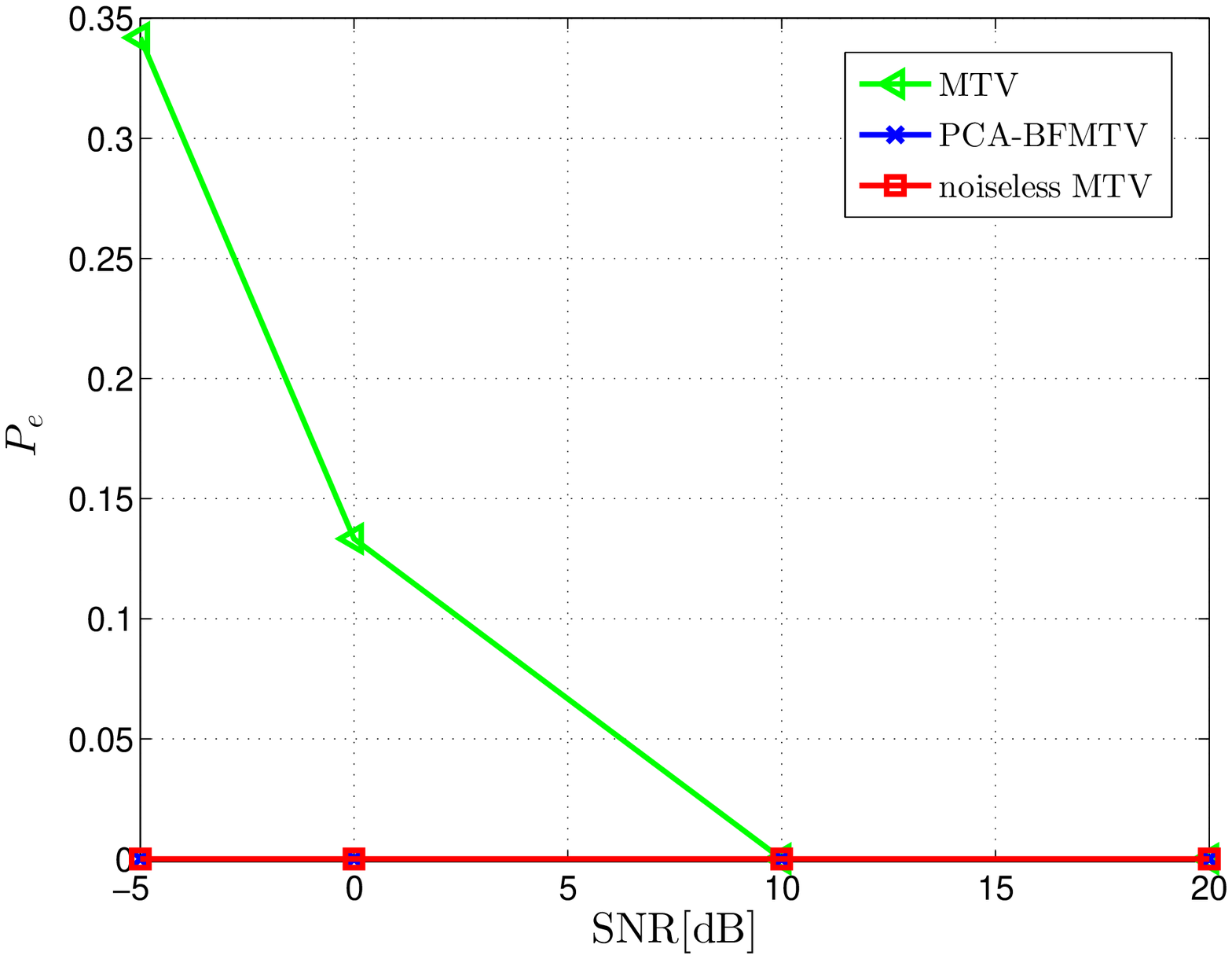}\\ (a) \\
\includegraphics[width=3.3in, height=2.1in]{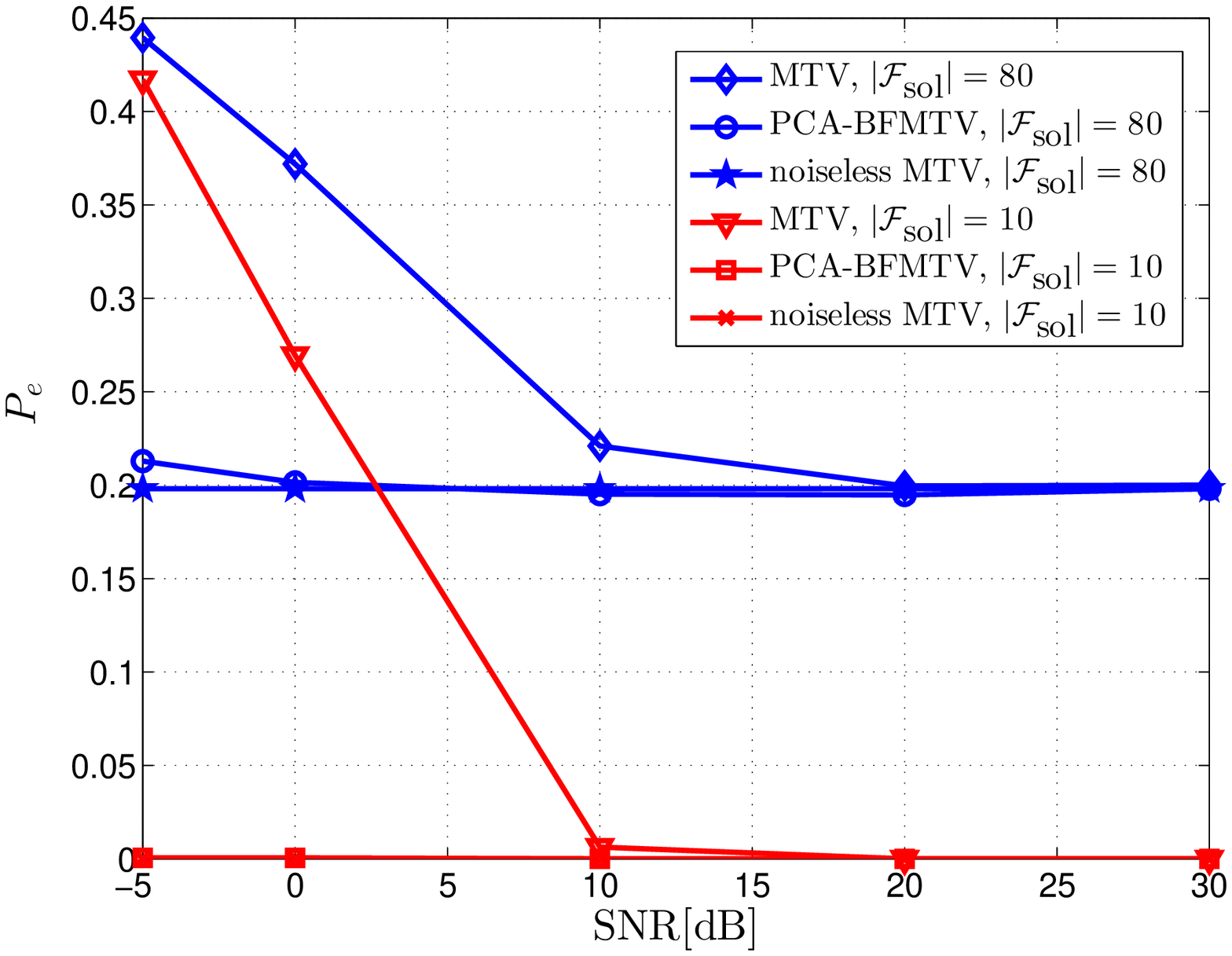}\\(b)
    \caption{Error probability by observing: $(\text{a})$  harmonic noisy signals; $(\text{b})$  harmonic plus solenoidal noisy signals.}
    \label{fig:allsimplex}
\end{figure}
To test the inference capabilities of the proposed methods,  in   Fig. \ref{fig:allsimplex}(a), we report the triangle error  probability $P_e$, defined as the percentage of incorrectly estimated triangles with respect to the number of cliques with three edges in the simplex, versus the signal-to-noise ratio ($\text{SNR}$), when the observation contains only harmonic flows plus noise.
We considered a simplex composed of $N=50$ nodes and with a percentage of filled triangles with respect to the number of second order cliques in the graph equal to $50\%$. We also  set  $M=50$,  $t^{*}=105$ and averaged our numerical results over $10^3$ zero-mean signal and noise random realizations.
The harmonic signal bandwidth $|\mathcal{F}_H|$ is chosen equal to $105$, which is equal to the dimension of the kernel of $\mL_1$.
From Fig. \ref{fig:allsimplex}(a), we can notice, as expected, that in the noiseless case the error probability is zero, since observing only harmonic flows  enables  perfect recovery of the matrix $\mB_2$. In the presence of noise, the MTV algorithm suffers and in fact we observe a non negligible error probability at low  $\text{SNR}$.  However, applying the PCA-BFMTV algorithm enables a significant recovery of performance, as evidenced by the blue curve that is entirely superimposed to the red curve, at least for the SNR values shown in the figure. In this example, the covariance matrix was estimated over   $10^{5}$ independent  observations of the edge signals. The optimal $\gamma$ coefficient was chosen after a cross validation operation following a line search approach aimed to minimize the  error probability.
The improvement of the PCA-BFMTV method with respect to the MTV method is due to the denoising made possible by the projection of the observed signal onto the space spanned by the  eigenvectors associated with the largest eigenvalues of the estimated covariance matrix.

To test the proposed  methods in the case where the observed signal contains both the solenoidal and  harmonic components, in Fig. \ref{fig:allsimplex}(b) we report  $P_e$ versus the $\text{SNR}$, for different values of the dimension of the subspace associated with the solenoidal part, indicated as $|{\cal F}_{\text{sol}}|$.
From Fig. \ref{fig:allsimplex}(b), we can observe that the performance of both algorithms MTV and PCA-BFMTV suffers when the bandwidth $|{\cal F}_{\text{sol}}|$ of the solenoidal component is large, whereas the performance degradation becomes negligible when $|{\cal F}_{\text{sol}}|$ is small. In all cases, PCA-BFMTV significantly outperforms the MTV algorithm, especially at low SNR values, because of its superior noise attenuation capabilities. As further numerical test, we run algorithm $\mathcal{P}^k_{\mathcal{S}}$  replacing the quadratic regularization term with the triple-wise coupling regularization function in  \eqref{eq:lovasz}, by obtaining the same performance of the PCA-BFMTV algorithm.\\
\noindent {\bf Performance on real data: } The real data set we used to test our algorithms is the set of mobile phone calls collected in the city of Milan, Italy, by  Telecom Italia, in the context of the Telecom Big Data Challenge \cite{TIM_data}.
The data are associated with a regular two-dimensional grid, composed of $100 \times 100$ points, superimposed to the city. Every point in the grid represents a square, of size $235$ meters. In particular, the data set collects the number $N_{ij}$ of calls from node $i$ to area $j$, as a function of time. There is an edge between nodes $i$ and $j$ only if there is a non null traffic between those points. The traffic has been aggregated in time, over time intervals of one hour. We define the flow signal over edge $(i, j)$ as $\Phi_{ij}=N_{ij}-N_{ji}$. We map all the values of matrix $\mathbf{\Phi}$ into a vector of flow signals $\bx^1$.
We observed the calls daily traffic during the month of December $2013$. The data are aggregated for each day over an interval of one hour.
Our first objective is to show whether there is an advantage in associating to the observed data set $\mX^1$ a complex of order $2$, i.e. a set of triangles, or it is sufficient to use a purely graph-based approach. In both cases, we rely on the same graph structure, whose $\mB_1$ comes from the data set, after an arbitrary choice of the edges' orientation. If we use a graph-based approach, we can build a basis of the observed flow signals using the eigenvectors of the so called edge Laplacian in \cite{mesbahi2010graph}, i.e. $\mL_1^{\text{low}}=\mB_1^T \mB_1$.
We call this basis $\mU_1^{\text{low}}$. As an alternative, our proposed approach is to build a basis using the eigenvectors of $\mL_1=\mB_1^T \mB_1+\mB_2 \mB_2^T$, where $\mB_2$ is estimated from the data set $\mX^1$ using our MTV algorithm. We call this basis $\mU_1$. To test the relative benefits of using $\mU_1$ as opposed to $\mU_1^{\text{low}}$, we run a basis pursuit algorithm with the goal of finding a good trade-off between the sparsity of the representation and the fitting error. More specifically, for any given observed vector $\bx^1(m)$, we look for the sparse vector $\bs^1$ as solution of the following basis pursuit problem
\cite{Donoho98}:
\beq \label{eq:bas_pur}
\begin{array}{lll}
 \underset{{\bs}^1 \in \mathbb{R}^E}{\text{min}} & \parallel
{\bs}^1\parallel_1   \qquad \qquad (\mathcal{B})\\
 \; \; \text{s.t.} & \parallel
 {\bx}^1 -\mV {\bs}^1\parallel_F \leq \epsilon
 \end{array}
\eeq
 where $\mV=\mU_1$ in our case, while $\mV=\mU_1^{\text{low}}$ in the graph-based approach. As a numerical result, in Fig. \ref{fig:spars} we report the sparsity of the recovered edge signals versus the mean estimation error $\parallel
 {\bx}^1 -\mV {\bs}^1\parallel_F$  considering as signal dictionary $\mV$ the eigenvectors of either the first-order Laplacian or the lower Laplacian.
 We used the MTV algorithm to infer the upper Laplacian matrix by setting the number $t^*$ of triangles   that we may detect equal to $800$.  This value is derived numerically through cross-validation over a training data set, by choosing the value of $t^*$ that yields the minimum norm  $\parallel {\bs}^1\parallel_1$ As can be observed from Fig. \ref{fig:spars},   using  the set of the eigenvectors of $\mL_1$ yields a much smaller MSE, for a given sparsity or, conversely, a much more sparse representation, for a given MSE. An intuitive reason why our method performs so much better than a purely graph-based approach is that the matrix $\mL_1$ has a much reduced kernel space with respect to $\mL_1^{\text{low}}$ and the basis built on $\mL_1$ captures much better some inner structure present in the data by inferring the structure of the additional term $\mB_2$ from the data itself.

As a further test, we tested the two basis $\mU_1$ and $\mU_1^{\text{low}}$ in terms of the capability to recover the entire flow signal from a subset of samples. To this end, we exploit the  band-limited property enforced by the sparse representation, enabling the use of the theory developed in Section V.A. Starting with the representation of each input vector $\bx^1$ as ${\bx}^1 =\mV {\bs}^1$, with either $\mV=\mU_1$ in our case, or $\mV=\mU_1^{\text{low}}$ in the graph-based approach, we used the Max-Det greedy sampling strategy in \cite{tsitsvero2016signals} to select the subset of edges where to sample the flow signal and then we used the recovery rule in (\ref{eq:recov_s1}) to retrieve the overall flow signal from the samples. The numerical results are reported in Fig.  \ref{fig:rec_err1}, representing the normalized recovering error of the edge signal versus the number $N_s$ of samples used to reconstruct the overall  signal. We can  notice how introducing the term $\mB_2 \mB_2^T$, we can achieve a much smaller error, for the same number of samples. \vspace{-0.3cm}

\begin{figure}[t!]
\centering
\includegraphics[width=8.7cm,height=5.7cm]{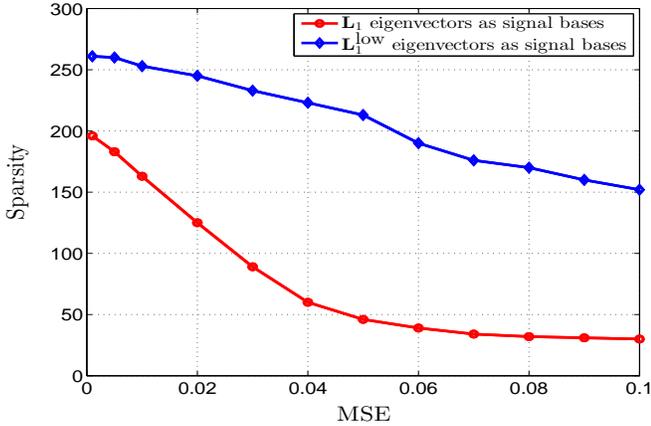}
\caption{Sparsity versus squared error.}
\label{fig:spars}
\end{figure}
\begin{figure}[t]
\centering
\includegraphics[width=8.7cm,height=5.8cm]{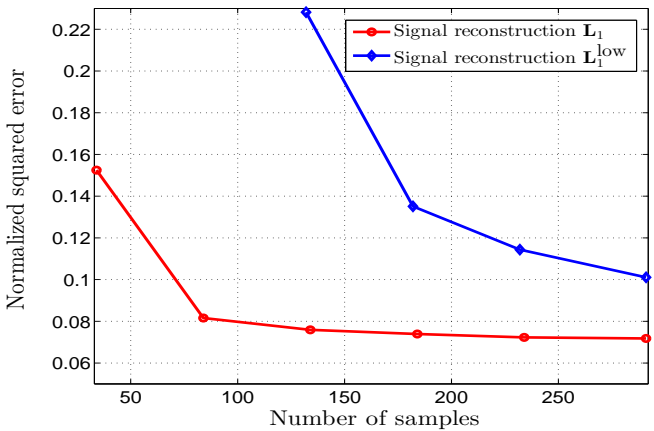}
\caption{Reconstruction error versus number of samples.}
\label{fig:rec_err1}
\end{figure}
\section{Conclusion}
\label{sec:conclusions}
In this paper we have presented an algebraic framework to analyze signals residing over a simplicial complex.  In particular, we focused on signals defined over the edges of a complex of order two, i.e. including triangles, and we showed that exploiting the full algebraic structure of this complex provides an advantage with respect to graph-based tools only. We have not analyzed higher order signals. Nevertheless, looking at the structure of the higher order Laplacian and to the Hodge decomposition, the tools derived in this paper can be directly translated to analyze signals defined over higher order structures. What would be missing in the higher order cases  would be a visual interpretation of the properties of the higher order Laplacian eigenvectors, as we could not be talking about solenoidal or irrotational behaviors.

We proposed a method to infer the structure of a second order simplicial complex from flow data and
we showed that, in applications over real wireless traffic data, the proposed approach can significantly outperform methods based only on graph representations. Furthermore, we proposed a method to analyze discrete vector fields and showed an application to the recovery of the RNA velocity field to predict the evolution of living cells. In such a case, using the eigenvectors of $\mathbf{L}_1$ we have been able to highlight irrotational and solenoidal behaviors that would have been difficult to highlight using only the eigenvectors of $\mathbf{L}_0$. Further developments should include both theoretical aspects, especially in the statistical modeling of random variables defined over a simplicial complex, and the generalization to higher order structures.
\vspace{-0.1cm}
\appendices{}

\section{Proof of Theorem 1}
We begin by briefly recalling the basic properties of Lov\'{a}sz extension \cite{Bach2013}, \cite{Hein2013} and then we proceed to the proof of Theorem 1.
\begin{definition} \label{def_lov}
Given a set ${\cal A}$, of cardinality $A=|\mathcal{A}|$, and its power set $2^{\cal A}$, let us consider a set function $Q: 2^{\mathcal{A}} \rightarrow \mathbb{R}$, with $Q(\emptyset)=0$.
Let $\bx \in \mathbb{R}^{A}$ be a vector of real variables, ordered w.l.o.g. in increasing order, such that $x_1\leq x_2 \leq \ldots \leq x_A$. Define
$C_0 \triangleq \mathcal{A}$  and $C_i \triangleq \{j \in \mathcal{A} : x_j>x_i\}$ for $i>0$.
Then, the Lov\'{a}sz extension $f: \mathbb{R}^{A} \rightarrow \mathbb{R}$
of $Q$, evaluated at $\bx$, is given by \cite{Bach2013}:
\beq \label{Lov_def}
\begin{split}
f(\bx)\,=& \,  Q(\mathcal{A}) x_1+ \ds \sum_{i=1}^{A-1} Q(C_{i})(x_{i+1}-x_i).
\end{split}
\eeq
\end{definition}
An interesting property of set functions is submodularity, defined as follows:
\begin{definition}
A set function  $Q : 2^{\mathcal{A}}\rightarrow \mathbb{R}$  is {\it submodular} if and only if, $\forall \mathcal{B},\mathcal{C}\subseteq \mathcal{A}$, it satisfies the following inequality:
\beq \nonumber
Q(\mathcal{B})+Q(\mathcal{C})\geq Q(\mathcal{B} \cup \mathcal{C})+Q(\mathcal{B} \cap \mathcal{C}).
\eeq
\end{definition}
A fundamental property  is that  a set function $Q$ is submodular iff its Lov\'{a}sz extension $f(\bx)$ is a convex function and that minimizing $f(\bx)$ on $[0,1]^A$, is equivalent to minimizing the set function $Q$ on $2^\mathcal{A}$  \cite[p.172]{Bach2013}.
 Now, we wish to apply Lov\'{a}sz extension to a set function, defined over the edge set $\mathcal{E}$, counting the number of triangles gluing nodes belonging to a tri-partition of the node set $\mathcal{V}$. To this end, let us consider a tri-partition $\{\A_0, \A_1, \A_2\}$ of the node set ${\mathcal{V}}$.
Extending the approach used for graphs, where  the cut-size is introduced as a set function  evaluated on the node set,   to define a triangle-cut size,  we need to introduce a set function  defined on the edge set.
Given a simplex $\mathcal{X}$, assume the faces of $\mathcal{X}$ to be positively oriented with the order of their vertices, i.e. for the
 $k$-simplex $\sigma^{k}_i$ we have $\sigma^{k}_i=[\sigma^{k-1}_{i_0}, \sigma^{k-1}_{i_1}, \ldots, \sigma^{k-1}_{i_{k}}]$ with $i_0<i_1<\ldots<i_k$. Given any non-empty tri-partition  $\{\mathcal{A}_0,\mathcal{A}_1,\mathcal{A}_2\}$  of the node set $\mathcal{V}$,
we define
 \beq \tau(\mathcal{A}_0,\mathcal{A}_1,\mathcal{A}_2)\!=\!\{\sigma^{2}_i \in \mathcal{T} : \, |\sigma^{2}_i\cap \mathcal{A}_j |=1, \, 0\leq j\leq 2\}
 \eeq
as the set of triangles $\sigma^2_i \in \mathcal{T}$ with exactly one vertex in each set $\mathcal{A}_j$, for $j=0,1,2$.
Given any $\sigma^2_i \in \tau(\mathcal{A}_0,\mathcal{A}_1,\mathcal{A}_2)$, there exists a permutation $\pi_i$ of its vertices $\sigma^{0}_{i_j}$, $j=0,1,2$,  such that $\sigma^{2}_i=[\sigma^0_{i_0},\sigma^0_{i_1},\sigma^0_{i_2}]$ and $\sigma^0_{i_j} \in \mathcal{A}_j$, for $j=0,1,2$.
 This returns  a linear ordering\footnote{A binary relation $\prec$ on a set $X$ is a \textit{linear} ordering on $X$ if it is transitive, antisymmetric and for any two elements $x,y \in X$, either $x\prec y$ or $y \prec x$.} of the vertices of the partition such that for all $i<j$, if $\sigma^0_k \in \mathcal{A}_i$ and $\sigma^0_m \in \mathcal{A}_j$, then  $\sigma^0_k$ precedes
 $\sigma^0_m$ in the vertex ordering.
  Given the tri-partition $\{\mathcal{A}_0,\mathcal{A}_1,\mathcal{A}_2\}$, we define the set $\mathcal{E}_T$ of oriented edges such that an oriented edge $[\sigma^0_i,\sigma^0_j] \in \mathcal{E}_T$ if and only if $\sigma^0_i \in \mathcal{A}_n$ and $\sigma^0_j \in \mathcal{A}_m$ for some $n<m$. Thus, we can think of the triangle cut size $F_1$ in \eqref{cut_trian}, equivalently,  as a different function $G(\mathcal{E}_T,\overline{\mathcal{E}}_T)$ defined on sets of oriented edges, where $\overline{\mathcal{E}}_T$ denotes the complement set of $\mathcal{E}_T$ in $\mathcal{E}$.
Given an orientation on the simplicial complex, we can associate  with the set of oriented edges $\mathcal{E}_T$ the vector $\mathbf{1}_{\mathcal{E}_T} \in \mathbb{R}^{E}$,  whose entries are $0$ for edges in
$\overline{\mathcal{E}}_T$ and $\pm 1$ for edges in $\mathcal{E}_T$, depending on the edge orientation.
It is  straightforward to check that $\mB_2^T \mathbf{1}_{\mathcal{E}_T}$, where $\mB_2$  is the edge-triangle incidence matrix defined in (\ref{inc_coeff}),  computes the triangle cut for  $\mathcal{E}_T$ in the sense that each entry of this vector, corresponding to a triangle, is  nonzero if and only if every vertex of that triangle is in a different set of the partition. As an example, considering the simplicial complex in Fig. \ref{cut}, we have $\mathcal{E}_T=\{e_2,e_3,e_4,e_7,e_8\}$  and the associated vector $\mathbf{1}_{\mathcal{E}_T}$ is $[0, 1, 1, 1, 0, 0, 1, 1, 0, 0, 0]^T$.
Then, using (\ref{B2_ex}),  we get $\mB_2^T \mathbf{1}_{\mathcal{E}_T}=[0,1,0]^T$ and the triangle cut size is equal to $1$.

Then, introducing the vector $\mathbf{t}_{\mathcal{E}_T}:= \mB_2^T \mathbf{1}_{\mathcal{E}_T}$, we define the edge set function as
\beq \label{G_A}
G(\mathcal{E}_T,\overline{\mathcal{E}}_T)= \sum_{j=1}^{T} |\mathbf{t}_{\mathcal{E}_T}(j)|.
\eeq
It can be proved that  $G$  is invariant  under any  permutation of the vertexes of the triangles, since the effect of any vertex permutation is a change of sign in the entries of $\mathbf{t}_{\mathcal{E}_T}$. Therefore, we get
  \beq \label{eq:G_F1}
  G(\mathcal{E}_T,\overline{\mathcal{E}}_T)=|\tau(\A_0, \A_1, \A_2)|=F_1(\mathcal{A}_0,\mathcal{A}_1,\mathcal{A}_2).
  \eeq

Now we wish to  derive the Lov\'{a}sz extension of the set function $G(\mathcal{E}_T,\overline{\mathcal{E}}_T)$.
Denote with  $\sigma^{2}_{i}=[\sigma^{1}_{i_0},\sigma^{1}_{i_1},\sigma^{1}_{i_2}]$ the oriented $2$-order simplex  where $1 \leq i_k\leq E$, for $k=0,1,2$.  To make explicit the dependence of  the function $G(\mathcal{E}_T,\overline{\mathcal{E}}_T)$
on the edge indexes  $i_0,i_1,i_2$,  we rewrite  (\ref{G_A}) as
 \beq \label{G_2}
G(\mathcal{E}_T,\overline{\mathcal{E}}_T)= \!\!\ds \sum_{\sigma^2_{i} \in \mathcal{T}} \left|\, \tilde{G}_{i}(\mathcal{E}_T,\overline{\mathcal{E}}_T) \right|
 \eeq
where, exploiting the structure of $\mB_2$, given in  (\ref{inc_coeff}),
each term $\tilde{G}_{i}$ can be written as (we omit the dependence of $\tilde{G}_{i}$ on the edge set partition,  for notation simplicity)
\beq \label{G_njk}
 \begin{split}
 \tilde{G}_{i}&=
 \mathbf{1}_{\mathcal{E}_T}(\sigma^{1}_{i_0})-\mathbf{1}_{\mathcal{E}_T}(\sigma^{1}_{i_1})+ \mathbf{1}_{\mathcal{E}_T}(\sigma^{1}_{i_2}),
 \end{split}
 \eeq
where $\mathbf{1}_{\mathcal{E}_T}(\sigma^{1}_{i_k})$, $k=0,1,2$, are the entries of $\mathbf{1}_{\mathcal{E}_T}$ corresponding to the edge $\sigma^{1}_{i_k}$.
$\tilde{G}_{i}$ is then a set function defined only on the power set of $\{i_0,i_1,i_2\}$. We can now derive
 its Lov\'{a}sz extension $f(x_{i_0},x_{i_1},x_{i_2})$.
 First,  assume $x_{i_0}\leq x_{i_1}\leq x_{i_2}$.
  Then, from Def. \ref{def_lov}, we have   $C_0=\{i_0,i_1,i_2\}$, $C_1=\{i_1,i_2\}$ and $C_2=\{i_2\}$.
   Therefore, from (\ref{G_njk}), it holds:
\beq
\begin{array}{lll}
\tilde{G}_{i}(C_0)=\mathbf{1}_{\mathcal{E}_T}(\sigma^1_{i_0})-\mathbf{1}_{\mathcal{E}_T}(\sigma^1_{i_1})+   \mathbf{1}_{\mathcal{E}_T}(\sigma^1_{i_2})=1 \\
\tilde{G}_{i}(C_1)=-\mathbf{1}_{\mathcal{E}_T}(\sigma^1_{i_1})+\mathbf{1}_{\mathcal{E}_T}(\sigma^1_{i_2})=0 \\
\tilde{G}_{i}(C_2)=\mathbf{1}_{\mathcal{E}_T}(\sigma^1_{i_2})=1.
\end{array}
\eeq
Then, from (\ref{Lov_def}) we get:
\beq \nonumber
\begin{split}
\!\!\!f(x_{i_0},x_{i_1},x_{i_2})\,= &\,  \tilde{G}_{i}(C_0) x_{i_0} +\tilde{G}_{i}(C_1) (x_{i_1}-x_{i_0}) \medskip \\ & +\tilde{G}_{i}(C_2) (x_{i_2}-x_{i_1})= x_{i_0}-x_{i_1}+x_{i_2}.
\end{split}
\eeq
 Let us now assume $x_{i_1}\leq x_{i_0}\leq x_{i_2}$.
 Thereby,  we have $C_0=\{i_0,i_1,i_2\}$, $C_1=\{i_0,i_2\}$ and $C_2=\{i_2\}$, so that it  results:
\beq
\begin{array}{lll}
\tilde{G}_{i}(C_0)=\mathbf{1}_{\mathcal{E}_T}(\sigma^1_{i_0})-\mathbf{1}_{\mathcal{E}_T}(\sigma^1_{i_1})+ \mathbf{1}_{\mathcal{E}_T}(\sigma^1_{i_2})=1 \\
\tilde{G}_{i}(C_1)=\mathbf{1}_{\mathcal{E}_T}(\sigma^1_{i_0})+\mathbf{1}_{\mathcal{E}_T}(\sigma^1_{i_2})=2 \\
\tilde{G}_{i}(C_2)=\mathbf{1}_{\mathcal{E}_T}(\sigma^1_{i_2})=1
\end{array}
\eeq
and \vspace{-0.2cm}
\beq \nonumber
f(x_{i_0},x_{i_1},x_{i_2})=2(x_{i_0}-x_{i_1})+(x_{i_2}-x_{i_0})+x_{i_1}=x_{i_0}-x_{i_1}+x_{i_2}.
\eeq
By following similar derivations, it is not difficult to show that for  $x_{i_0}\leq x_{i_2} \leq x_{i_1}$, $x_{i_1}\leq x_{i_2} \leq x_{i_0}$,
$x_{i_2}\leq x_{i_1} \leq x_{i_0}$ and $x_{i_2}\leq x_{i_0} \leq x_{i_1}$, it holds $f(x_{i_0},x_{i_1},x_{i_2})\,= x_{i_0}-x_{i_1}+x_{i_2}$.
Therefore, from (\ref{G_2}), defining the edge signal $\bx^1=[x_1,\ldots, x_E]^T$,
we can write the Lov\'{a}sz extension of $G(\mathcal{E}_T,\bar{\mathcal{E}_T})$  exploiting the additive property  \cite{Bach2013} as
\beq \vspace{-0.1cm}
f_1(\bx^1)= \!\!\ds \sum_{\sigma^2_{i} \in \mathcal{T}} \left| x_{i_0}^1-x_{i_1}^1+x_{i_2}^1 \right|
\eeq
or, equivalently, as
\beq \label{fx_abs}
f_1(\bx^1)= \ds \sum_{j=1}^{T} \left|\sum_{i=1}^{E} B_2(i,j) x_i^1\right|.
\eeq
From the equality
$G(\mathcal{E}_T,\bar{\mathcal{E}_T})=F_1(\mathcal{A}_0,\mathcal{A}_1,\mathcal{A}_2)$ in (\ref{eq:G_F1}), we can state that $f_1(\bx^1)$ is the
Lov\'{a}sz extension of $F_1$.
 This completes the proof of Theorem $1$. \vspace{-0.2cm}


\bibliographystyle{IEEEbib}
\bibliography{main}
\clearpage
\newpage
\section{Supporting material}
This document contains some supporting materials complementing the paper:
``Topological Signal Processing over Simplicial Complexes''  to  appear in IEEE Transactions on Signal Processing. Section A contains Proposition $2$,  and its proof, that will be instrumental in proving Theorems $3$
and $4$, respectively, in Sections B and C.
\subsection{Proposition 2}
We need first to derive the relationship between the bandwidth of the edge signal $\bs^{1}_{\text{irr}}=\mB_1^T \bs^{0}$ and that of
the vertex signal $\bs^{0}$, as
stated in the following proposition.
\begin{proposition} \label{s_irr_band}
Let $\bs^{1}_{\text{irr}}=\mB_1^T \bs^{0}$ be the irrotational part of $\bs^1$ with $\bs^{0}=\mathbf{F}^{0}_{\mathcal{F}_0} \bs^0$ a \textcolor{black}{$\mathcal{F}_0$-bandlimited vertex signal. Then, $\bs^{1}_{\text{irr}}$ is a $\mathcal{F}_{\text{irr}}$-bandlimited edge signal with $\mathcal{F}_{\text{irr}}$ the set of indexes of the eigenvectors of $\mL_1$ stacked in the columns of $\mU_{\mathcal{F}_{\text{irr}}}$,  such that $\mU_{\mathcal{F}_{\text{irr}}}=\mB_1^T \mU^0_{\mathcal{F}_{0}-\mathcal{C}_{1}}$. The bandwidth of $\bs^{1}_{\text{irr}}$ is $|\mathcal{F}_{\text{irr}}|=|\mathcal{F}_0|-c_1$, where $c_1 \geq 0$ is the number of eigenvectors in the bandwidth of $\bs^{0}$ belonging to $\text{ker}(\mL_0)$.}
\end{proposition}
\begin{proof}
Let us define  $\mU_{\mathcal{F}_0}^{0}$ the $N\times |\mathcal{F}_0|$ matrix whose columns $\bu_i^{0}$,  $ \forall i \in  \mathcal{F}_0$ are the eigenvectors of the 0-order Laplacian $\mB_1 \mB_1^T$.
Since $\bs^{1}_{\text{irr}}=\mB_1^T \bs^{0}$ we get
\beq \label{eq:s_1s0}
\bs^{1}_{\text{irr}}=\mB_1^T \bs^{0}=\mB_1^T \mU_{\mathcal{F}_0}^{0} (\mU_{\mathcal{F}_0}^{0})^T \bs^{0}
\eeq
where the last equality follows from the bandlimitedness of $\bs^{0}$, i.e. $\bs^{0}=\mU_{\mathcal{F}_0}^{0} (\mU_{\mathcal{F}_0}^{0})^T \bs^{0}$.
From  the property 2) in Prop. 1, at each eigenvector $\bu_i^{0}$   of $\mB_1 \mB_1^T$ with $\bu_i^{0} \notin \text{ker}(\mB_1 \mB_1^T)$ corresponds an eigenvector $\mB_1^T \bu_i^{0}$ of $\mB_1^T \mB_1$ with the same eigenvalue. \textcolor{black}{Let $\mU_{\mathcal{F}_{\text{irr}}}$ denote the $E\times |\mathcal{F}_{\text{irr}}|$ matrix whose columns are the eigenvectors of $\mL_1$ associated to the frequency index set $\mathcal{F}_{\text{irr}}$ with
\beq \label{eq:F_irr}
\mathcal{F}_{\text{irr}}=\{ \ell_i \in \mathcal{F}  :  \bu_{\ell_i}=\mB_1^T \bu^0_i,  \bu_i^{0} \notin \text{ker}(\mB_1 \mB_1^T), \forall i \in \mathcal{F}_{0}-\mathcal{C}_{1} \}.\eeq}
Then,
if  $\mU_{\mathcal{F}_0}^{0}=[\mU_{\mathcal{C}_1}^{0}, \mU_{\mathcal{F}_0-\mathcal{C}_1}^{0}]$,
we get
\beq
\mB_1^T \mU_{\mathcal{F}_0}^{0}=[\mathbf{O}_{\mathcal{C}_1}, \, \mU_{\mathcal{F}_{\text{irr}}}].
\eeq
Therefore, equation (\ref{eq:s_1s0}) reduces to
\beq \label{eq:s12}
\bs^{1}_{\text{irr}}=\mU_{\textcolor{black}{\mathcal{F}_{\text{irr}}}} (\mU_{\mathcal{F}_0-\mathcal{C}_1}^{0})^T \bs^{0}.
\eeq
From the equality $\mB_1^T\mB_1 \bu_{\ell_i}=\lambda_{\ell_i} \bu_{\ell_i}$, multiplying
both sides by $\mB_1$, we easily derive $\bu^{0}_i=\mB_1 \bu_{\ell_i}$, $\forall \bu_{\ell_i} \notin \text{ker}(\mB_1^T \mB_1)$ so that equation (\ref{eq:s12}) is equivalent to
\beq \label{eq:s13}
\bs^{1}_{\text{irr}}=\mU_{\textcolor{black}{\mathcal{F}_{\text{irr}}}} \mU_{\textcolor{black}{\mathcal{F}_{\text{irr}}}}^{T}\mB_1^T \bs^{0}.
\eeq
Since $\bs^{1}_{\text{irr}}=\mB_1^T \bs^{0}$,  we can rewrite (\ref{eq:s13}) as
\beq
\bs^{1}_{\text{irr}}=\mU_{\mathcal{F}_{\text{irr}}} \mU_{\mathcal{F}_{\text{irr}}}^{T} \bs^{1}_{\text{irr}}.
\eeq
This last equality proves that $\bs^{1}_{\text{irr}}$ is a \textcolor{black}{$\mathcal{F}_{\text{irr}}$-bandlimited edge signal with $\mathcal{F}_{\text{irr}}$ as in (\ref{eq:F_irr}) and bandwidth $|\mathcal{F}_{\text{irr}}|=|\mathcal{F}_0|-c_1$, defining $c_1=|\mathcal{C}_1|\geq 0$.}
\end{proof}
\subsection{Proof of Theorem 3}
\label{Proof of Theorem 3}
Given the sampled signals $\bs^{1}_{\mathcal{S}}$ and $\bs^{0}_{\mathcal{A}}$, we get
\beq
\begin{array}{lll}
\bs^{0}_{\mathcal{A}}=\mD_{\mathcal{A}} \bs^0\\
\bs^{1}_{\mathcal{S}}=\mD_{\mathcal{S}} \bs^1=\mD_{\mathcal{S}} \bar{\bs}^1+ \mD_{\mathcal{S}} \mB_1^T\bs^0
\end{array}
\eeq
with $\bar{\bs}^1=\bs^1_{\text{sol}}+\bs^1_{{H}}$.
Then, since it holds $\bs^0=\mF_{\mathcal{F}_0}^0\bs^0$ and $\bs^1=\mF_{\mathcal{F}}\bs^1$, we easily obtain
\beq
\left[\begin{array}{lll}
\bs^{0}_{\mathcal{A}} \medskip\\
\bs^{1}_{\mathcal{S}}
\end{array}\right]=\mG \left[\begin{array}{lll}
\bs^{0} \medskip\\
\bar{\bs}^{1}
\end{array}\right],
\eeq
where
\beq
\begin{array}{lll}
\mG& =\left[ \begin{array}{lll}\mD_{\mathcal{A}} \mF_{\mathcal{F}_0}^0 & \mathbf{O}\medskip\\
\mD_{\mathcal{S}} \mB_1^T \mF_{\mathcal{F}_0}^0 &
 \mD_{\mathcal{S}} \mF_{\mathcal{F}_{\text{sH}}}\end{array}\right] \medskip \\ &= \left[ \begin{array}{lll}(\mI-\overline{\mD}_{\mathcal{A}} \mF_{\mathcal{F}_0}^0)  & \mathbf{O}\medskip\\
(\mI-\overline{\mD}_{\mathcal{S}}) \mB_1^T \mF_{\mathcal{F}_0}^0 &
 (\mI-\overline{\mD}_{\mathcal{S}} \mF_{\mathcal{F}_{\text{sH}}}) \end{array}\right]
 \end{array}
\eeq
with $\mathcal{F}_{\text{sH}}$ the set of frequency indexes in $\mathcal{F}$ corresponding to the eigenvectors of $\mL_1$ belonging to the solenoidal and harmonic subspaces.
Assuming that both the conditions
$\|\bar{\mD}_{\mathcal{A}}\mF_{\mathcal{F}_0}^0\|_2<1$ and $\|\bar{\mD}_{\mathcal{S}}\mF_{\mathcal{F}_{\text{sH}}}\|_2<1$ hold true, the matrix $\mG$ is invertible and  from the inverse of partitioned matrices \cite{Bernstein}, we get
\beq
\begin{array}{lll}
\mQ=\mG^{-1}=\left[ \begin{array}{lll}(\mI-\overline{\mD}_{\mathcal{A}} \mF_{\mathcal{F}_0}^0)^{-1}  & \mathbf{O}\medskip\\
 \mP &
 (\mI-\overline{\mD}_{\mathcal{S}} \mF_{\mathcal{F}_{\text{sH}}})^{-1} \end{array}\right]\end{array}
\eeq
with $\mP=- (\mI-\overline{\mD}_{\mathcal{S}} \mF_{\mathcal{F}_{\text{sH}}})^{-1}{\mD}_{\mathcal{S}} \mB_1^T \mF_{\mathcal{F}_0}^0 (\mI-\overline{\mD}_{\mathcal{A}} \mF_{\mathcal{F}_0}^0)^{-1}$.
Then we can recovery the signals $\bs^0$ and $\bar{\bs}^{1}$ as
\beq \label{eq:recover11}
\left[\begin{array}{lll}
\bs^{0} \medskip\\
\bar{\bs}^{1}
\end{array}\right]=\mQ \left[\begin{array}{lll}
\bs^{0}_{\mathcal{A}} \medskip\\
\bs^{1}_{\mathcal{S}}
\end{array}\right].
\eeq
This concludes the proof of point a). Let us  prove next point b).
From Proposition $2$ it results that, if $\bs^0$ is  $\mathcal{F}_0$-bandlimited, then $\bs^{1}_{\text{irr}}=\mB_1^T \bs^{0}$ is a
\textcolor{black}{$\mathcal{F}_{\text{irr}}$-bandlimted  signal with $\mathcal{F}_{\text{irr}}$ given  in (\ref{eq:F_irr})}. This implies that the edge signal $\bs^1=\bs^{1}_{\text{sol}}+\bs^{1}_{{H}}+\bs^{1}_{\text{irr}}$ is a \textcolor{black}{$\mathcal{F}$-bandlimited edge signal with $\mathcal{F}=\mathcal{F}_{\text{sH}} \cup \mathcal{F}_{\text{irr}}$} and bandwidth  $|\mathcal{F}|=|\mathcal{F}_{\text{sH}}|+ |\mathcal{F}_0|-c_1$.
Let us now consider the system in (\ref{eq:recover11}). We  get
\beq \label{eq:recover2}
\begin{array}{lll}
\bs^{0} &=(\mI-\overline{\mD}_{\mathcal{A}} \mF_{\mathcal{F}_0}^0)^{-1} \bs^{0}_{\mathcal{A}}\\
\bar{\bs}^{1}& =-(\mI-\overline{\mD}_{\mathcal{S}} \mF_{\mathcal{F}_{\text{sH}}})^{-1} \mD_{\mathcal{S}} \mB_1^T \mF_{\mathcal{F}_0}^0  (\mI-\overline{\mD}_{\mathcal{A}} \mF_{\mathcal{F}_0}^0)^{-1} \bs^{0}_{\mathcal{A}}+\medskip \\ & \quad (\mI-\overline{\mD}_{\mathcal{S}} \mF_{\mathcal{F}_{\text{sH}}})^{-1}\bs^{1}_{\mathcal{S}} .
\end{array}
\eeq
Using the first equation in (\ref{eq:recover2}) and the fact that $\mD_{\mathcal{S}} \bs^{1}=
\mD_{\mathcal{S}} \bar{\bs}^{1}+\mD_{\mathcal{S}} \mB_{1}^{T} \bs^{0}$, it holds
\beq \label{eq:s_solband}
\begin{array}{lll}
\bar{\bs}^{1}& =(\mI-\overline{\mD}_{\mathcal{S}} \mF_{\mathcal{F}_{\text{sH}}})^{-1} (\mD_{\mathcal{S}} \bar{\bs}^{1} + \mD_{\mathcal{S}} \mB_1^{T}\bs^{0}-\mD_{\mathcal{S}} \mB_1^{T}\mF_{\mathcal{F}_0}^0 \bs^{0}) \medskip\\
&=(\mI-\overline{\mD}_{\mathcal{S}} \mF_{\mathcal{F}_{\text{sH}}})^{-1} \mD_{\mathcal{S}} \bar{\bs}^{1}
\end{array}
\eeq
where the last equality follows from  $\bs^0=\mF_{\mathcal{F}_0}^0 \bs^0$. Hence, from  (\ref{eq:s_solband}), it follows that to perfectly recover the solenoidal and harmonic parts of $\bs^{1}$ we need a number of samples $N_1$ at least  equal to the signal bandwidth $|\mathcal{F}_{sH}|$. Finally, from   the first equation in (\ref{eq:recover2}) it follows that to perfectly recovering $\bs^0$ we need a number of samples $N_0$ at least  equal to the bandwidth $|\mathcal{F}_0|$ of the vertex signal $\bs^0$. This concludes the proof of point b) in the theorem.
\subsection{Proof of Theorem 4}
\label{Proof of Theorem 4}
Given the sampled signals  $\bs^{0}_{\mathcal{A}}$, $\bs^{1}_{\mathcal{S}}$ and $\bs^{2}_{\mathcal{M}}$, we have
\beq
\begin{array}{lll}
\bs^{0}_{\mathcal{A}}=\mD_{\mathcal{A}} \bs^0 \medskip\\
\bs^{1}_{\mathcal{S}}= \, \mD_{\mathcal{S}} \bs^1=\mD_{\mathcal{S}} \mB_2 {\bs}^2+ \mD_{\mathcal{S}}  {\bs}^1_{H}+ \mD_{\mathcal{S}} \mB_1^T\bs^0  \medskip\\
\bs^{2}_{\mathcal{M}}=\mD_{\mathcal{M}} \mF_{\mathcal{F}_2}^2 \bs^2 .
\end{array}
\eeq
Then, using the bandlimitedness property, so that $\bs^0=\mF_{\mathcal{F}_0}^0 \bs^0$, $\bs^1_H=\mF_{\mathcal{F}_H} \bs^1_H$ and
$\bs^2=\mF_{\mathcal{F}_2}^2 \bs^2$, we get
\beq
\left[\begin{array}{lll}
\bs^{0}_{\mathcal{A}} \medskip\\
\bs^{1}_{\mathcal{S}}\medskip\\
\bs^{2}_{\mathcal{M}}
\end{array}\right]=\bar{\mG} \left[\begin{array}{lll}
\bs^{0} \medskip\\
{\bs}^{1}_H\medskip\\
{\bs}^{2}
\end{array}\right]
\eeq
with
\beq
\begin{array}{lll}
\bar{\mG}=\left[ \begin{array}{lll}{\mD}_{\mathcal{A}} \mF_{\mathcal{F}_0}^0  & \mathbf{O} & \mathbf{O} \medskip\\
 {\mD}_{\mathcal{S}}\mB_1^T \mF_{\mathcal{F}_0}^0  & {\mD}_{\mathcal{S}} \mF_{\mathcal{F}_H}
  & {\mD}_{\mathcal{S}} \mB_2 \mF_{\mathcal{F}_2}^2 \medskip\\ \mathbf{O} & \mathbf{O} & {\mD}_{\mathcal{M}} \mF_{\mathcal{F}_2}^2
 \end{array}\right]\end{array}.
 \eeq
Under the assumptions
$\|\bar{\mD}_{\mathcal{A}}\mF_{\mathcal{F}_0}^0\|_2<1$, $\|\bar{\mD}_{\mathcal{S}}\mF_{\mathcal{F}_{H}}\|_2<1$ and $\|\bar{\mD}_{\mathcal{M}}\mF_{\mathcal{F}_2}^2\|_2<1$, the matrix $\bar{\mG}$ becomes invertible and  from the inverse of partitioned matrices \cite{Bernstein}, we get
\beq
\begin{array}{lll}
\mR=\bar{\mG}^{-1}=\left[ \begin{array}{lll}(\mI-\overline{\mD}_{\mathcal{A}} \mF_{\mathcal{F}_0}^0)^{-1}  & \mathbf{O} & \mathbf{O} \medskip\\
 \mP_1 & \mathbf{P}_2
  & \mathbf{P}_3 \medskip\\ \mathbf{O} & \mathbf{O} & (\mI-\overline{\mD}_{\mathcal{M}} \mF_{\mathcal{F}_{2}}^2 )^{-1}
 \end{array}\right]\end{array}
\eeq
and $\mP_1=- (\mI-\overline{\mD}_{\mathcal{S}} \mF_{\mathcal{F}_{\text{H}}})^{-1}{\mD}_{\mathcal{S}} \mB_1^T \mF_{\mathcal{F}_0}^0 (\mI-\overline{\mD}_{\mathcal{A}} \mF_{\mathcal{F}_0}^0)^{-1}$, $\mP_2= (\mI-\overline{\mD}_{\mathcal{S}} \mF_{\mathcal{F}_{\text{H}}} )^{-1}$,
$\mP_3=- (\mI-\overline{\mD}_{\mathcal{S}} \mF_{\mathcal{F}_{\text{H}}})^{-1}{\mD}_{\mathcal{S}} \mB_2 \mF_{\mathcal{F}_2}^2 (\mI-\overline{\mD}_{\mathcal{M}} \mF_{\mathcal{F}_2}^2)^{-1}$.
Then we can recovery the signals $\bs^0$, ${\bs}_{H}^1$ and ${\bs}^{2}$ as
\beq \label{eq:recover}
\left[\begin{array}{lll}
\bs^{0} \medskip\\
{\bs}_{H}^1\medskip\\
{\bs}^{2}
\end{array}\right]=\mR \left[\begin{array}{lll}
\bs^{0}_{\mathcal{A}} \medskip\\
\bs^{1}_{\mathcal{S}}\medskip\\
\bs^{2}_{\mathcal{M}}
\end{array}\right].
\eeq
This concludes the proof of point a). Let us  prove next point b).
From Proposition \ref{s_irr_band},
 $\bs^{1}_{\text{irr}}$ is a \textcolor{black}{$\mathcal{F}_{\text{irr}}$-bandlimited  signal with $\mathcal{F}_{\text{irr}}$ given  in (\ref{eq:F_irr})}. To find the bandwidth of the solenoidal part $\bs^{1}_{\text{sol}}=\mB_2 \bs^2$ we can proceed in a similar way to the proof of Proposition 2. Using the bandlimitedness property of $\bs^2$, so that $\bs^2=\mF_{\mathcal{F}_2}^2 \bs^2$, it results
\beq \label{eq_s_sol}
\bs^{1}_{\text{sol}}=\mB_2 \mU_{\mathcal{F}_2}^2 \mU_{\mathcal{F}_2}^{2 \, T} \bs^2\eeq
with $\mU_{\mathcal{F}_2}^2 \in \mathbb{R}^{T \times |\mathcal{F}_2|}$
the matrix whose columns  $\bu_i^2$, $\forall i \in \mathcal{F}_2$ are the eigenvectors of the second-order Laplacian $\mL_2=\mB_2^T \mB_2$.
From Proposition 1
at each eigenvector  $\bu^2$ with $\bu^2 \notin \mbox{ker}(\mB_2^T \mB_2)$ corresponds an eigenvector $\mB_2 \bu^2$ of $\mB_2 \mB_2^T$
with the same eigenvalue.
\textcolor{black}{Denote with $\mU_{\mathcal{F}_{\text{sol}}}$  the $E\times |\mathcal{F}_{\text{sol}}|$ matrix whose columns are the eigenvectors of $\mL_1$ associated to the frequency index set $\mathcal{F}_{\text{sol}}$ with
\beq \label{eq:F_sol}
\mathcal{F}_{\text{sol}}=\{ \ell_i \in \mathcal{F} \, : \, \bu_{\ell_i}=\mB_2 \bu^2_i, \, \bu^{2}_i \notin \text{ker}(\mB_2^T \mB_2), \forall i \in \mathcal{F}_{2}-\mathcal{C}_{2} \}.\eeq}
Let us write $\mU_{\mathcal{F}_2}^{2}$ as
$$\mU_{\mathcal{F}_2}^{2}=[\mU_{\mathcal{C}_2}^{2}, \mU_{\mathcal{F}_2-\mathcal{C}_2}^{2}]$$ where the columns of $\mU_{\mathcal{C}_2}^{2}$ are the $c_2=|\mathcal{C}_2|\geq 0$ eigenvectors in the kernel of  $\mL_2$   belonging  to the bandwidth of $\bs^2$. Therefore, it results
\beq
\mB_2 \mU_{\mathcal{F}_2}^{2}=[\mathbf{O}_{\mathcal{C}_2}, \, \mU_{\textcolor{black}{\mathcal{F}_{\text{sol}}}}],
\eeq\\
where we used the equality $\mB_2 \mU_{\mathcal{F}_2-\mathcal{C}_2}^{2}=\mU_{\textcolor{black}{\mathcal{F}_{\text{sol}}}}$.
Then equation (\ref{eq_s_sol}) reduces to
\beq \label{eq:sol1}
\bs^{1}_{\text{sol}}=\mU_{\textcolor{black}{\mathcal{F}_{\text{sol}}}} (\mU_{{{\mathcal{F}_2-\mathcal{C}_2}}}^{2})^T \bs^{2}.
\eeq
Since it holds $\mB_2^T \bu_{\textcolor{black}{\ell_i}}=\bu_i^2$, $\forall \bu_i^2 \notin \mbox{ker}(\mB_2)$, we easily get
\beq
\bs^{1}_{\text{sol}}= \mU_{\textcolor{black}{\mathcal{F}_{\text{sol}}}} (\mU_{\textcolor{black}{\mathcal{F}_{\text{sol}}}})^T \mB_2 \bs^{2}=\mU_{\textcolor{black}{\mathcal{F}_{\text{sol}}}} (\mU_{\textcolor{black}{\mathcal{F}_{\text{sol}}}})^T \bs^{1}_{\text{sol}}.
\eeq
Then $\bs^{1}_{\text{sol}}$ is a \textcolor{black}{$\mathcal{F}_{\text{sol}}$-bandlimited edge signal with $\mathcal{F}_{\text{sol}}$ given in (\ref{eq:F_sol}). This implies that $\bs^1$ is $\mathcal{F}$-bandlimited with $\mathcal{F}=\mathcal{F}_{\text{sol}} \cup \mathcal{F}_{\text{H}} \cup \mathcal{F}_{\text{irr}}$} and bandwidth $|\mathcal{F}|=|\mathcal{F}_0|+|\mathcal{F}_{\text{H}}|+|\mathcal{F}_{2}|-(c_1+c_2)$.
Let us now consider the system in (\ref{eq:recover}). We  get
\beq \label{eq:recover1}
\begin{array}{lll}
\bs^{0} &=(\mI-\overline{\mD}_{\mathcal{A}} \mF_{\mathcal{F}_0}^0)^{-1} \bs^{0}_{\mathcal{A}}\medskip\\
{\bs}_{H}^1& =-(\mI-\overline{\mD}_{\mathcal{S}} \mF_{\mathcal{F}_{\text{H}}})^{-1} \mD_{\mathcal{S}} \mB_1^T \mF_{\mathcal{F}_0}^0  (\mI-\overline{\mD}_{\mathcal{A}} \mF_{\mathcal{F}_0}^0)^{-1} \bs^{0}_{\mathcal{A}}+\medskip \\ & \quad (\mI-\overline{\mD}_{\mathcal{S}} \mF_{\mathcal{F}_{\text{H}}})^{-1}\bs^{1}_{\mathcal{S}}-(\mI-\overline{\mD}_{\mathcal{S}} \mF_{\mathcal{F}_{\text{H}}})^{-1} \mD_{\mathcal{S}} \mB_2 \mF_{\mathcal{F}_2}^2 \medskip \\ & \quad \cdot (\mI-\overline{\mD}_{\mathcal{M}} \mF_{\mathcal{F}_{2}}^2)^{-1}\bs^2_{\mathcal{M}}\medskip \\
\bs^{2} &=(\mI-\overline{\mD}_{\mathcal{M}} \mF_{\mathcal{F}_2}^2)^{-1} \bs^{2}_{\mathcal{M}}.
\end{array}
\eeq
Using equations in (\ref{eq:recover1}) and the fact that $\bs^{1}_{\mathcal{S}}=
\mD_{\mathcal{S}}  \mB_2 {\bs}^{2}+\mD_{\mathcal{S}}  \mB_{1}^{T} \bs^{0} +\mD_{\mathcal{S}}  \bs_H^1$, it holds
\beq \label{eq:s_sol_band}
\begin{array}{lll}
{\bs}_{H}^1& =(\mI-\overline{\mD}_{\mathcal{S}} \mF_{\mathcal{F}_{H}})^{-1} \mD_{\mathcal{S}} {\bs}_{H}^1.
\end{array}
\eeq
 Hence, from  (\ref{eq:s_sol_band}), it follows that to perfectly recover  the harmonic part of $\bs^{1}$ we need a number of edge samples  $N_1$ at least  equal to the signal bandwidth $|\mathcal{F}_{H}|$.
 Finally, from the first and last equation in  (\ref{eq:recover1}), we  conclude that
to retrieve $\bs^0$ and  the solenoidal signal $\bs^{1}_{\text{sol}}$ we need, respectively,  a number of node samples $N_0\geq |\mathcal{F}_0|$  and a number of triangle signal samples $N_2 \geq |\mathcal{F}_{2}|$.

\end{document}